\documentclass[runningheads,a4paper]{llncs}

\usepackage{amsmath}
\usepackage{amssymb}
\usepackage[dvipsnames]{xcolor}
\usepackage{graphicx}
\usepackage{semantic}
\usepackage{mathpartir}
\usepackage[normalem]{ulem}
\usepackage{stmaryrd}
\usepackage{enumitem}
\usepackage{booktabs}

\usepackage{tikz}
\usetikzlibrary{arrows.meta,positioning,intersections,shapes.geometric,shapes.multipart,calc}

\usepackage{listings,listings-rust}
\usepackage[hidelinks]{hyperref}
\usepackage{cleveref}
\definecolor{Cerulean}{rgb}{0.20,0.31,0.0.90}
\definecolor{RedOrange}{rgb}{0.70,0.59,0.0.20}
\definecolor{ForestGreen}{rgb}{0.1,0.60,0.2}
\definecolor{SomeRed}{rgb}{0.80,0.3,0.0.20}
\definecolor{lightGray}{rgb}{0.85,0.85,0.86} 
\definecolor{lGray}{rgb}{0.31,0.31,0.33} 
\definecolor{lYellow}{rgb}{0.99,0.78,0.07} 
\definecolor{lila}{rgb}{0.7,0.30,0.77} 
\definecolor{cloudyblue}{rgb}{206,221,246}

\newcommand{\src}[1]{\color{RoyalBlue}{\mathsf{#1}}}
\newcommand{\trg}[1]{\color{RedOrange}{\mathbf{#1}}}

\definecolor{RedOrange}{RGB}{178,34,34}
\usepackage{ifthen}

\newboolean{showcomments}
\setboolean{showcomments}{false}

\makeatletter
\newcommand{\mynote}[3]{%
  \ifthenelse{\boolean{showcomments}}{%
   \fbox{\bfseries\sffamily\scriptsize#1}%
   {\small\textsf{\emph{\color{#3}{#2}}}}}%
  {%
   \@bsphack
   \@esphack
  }%
}
\makeatother

\definecolor{asparagus}{rgb}{0.53, 0.66, 0.42}

\newcommand{\out}[1]{}

\setboolean{showcomments}{false}

\usepackage{xspace}
\usepackage[dvipsnames]{xcolor}

\newcommand{\itree}{\text{ITree}\xspace}
\newcommand{\itrees}{\text{ITrees}\xspace}
\newcommand{\srcE}{\ensuremath{{\src E}}}
\newcommand{\trgE}{\ensuremath{{\trg E}}}

\newcommand{\srcRetv}[1]{\ensuremath{{\src {#1}}}}
\newcommand{\trgRetv}[1]{\ensuremath{{\trg {#1}}}}
\newcommand{\handler}[1]{\ensuremath{H_{\mathsf{#1}}}}
\newcommand{\srcH}{\ensuremath{{\src H}}}
\newcommand{\trgH}{\ensuremath{{\trg H}}}

\newcommand{\SIM}{\ensuremath{\mathsf{SIM}}}
\newcommand{\simIMPASM}{\ensuremath{\mathsf{Sim}}}

\newcommand{\simFunctor}{\ensuremath{\mathcal{F}_R}}

\newcommand{\iProp}{\mathit{iProp}}
\newcommand{\updateMod}{\mid \! \Rrightarrow}

\newcommand{\wpIMP}[2]{\ensuremath{\mathsf{WP}_{\mathsf{IMP}}(#1,#2)}}
\newcommand{\srcT}{{\src T}}
\newcommand{\trgT}{{\trg T}}
\newcommand{\srcGamma}{{\src \Gamma}}
\newcommand{\trgGamma}{{\trg \Gamma}}
\newcommand{\srcGint}{{\src \Gamma}_{\mathit{int}}}
\newcommand{\wpIMPx}[3]{\ensuremath{\mathsf{WP}^{#1}_{\mathsf{IMP}}(#2,#3)}}
\newcommand{\denoteIMP}[1]{\ensuremath{{\color{RoyalBlue}\mathcal{S}}(#1)}}

\newcommand{\denoteASM}[1]{\ensuremath{{\color{RedOrange}\mathcal{T}}(#1)}}

\newcommand{\compProp}[1]{\llparenthesis #1 \rrparenthesis}

\newcommand*{\wand}[0]{-\!\!*}

\newcommand{\itreeER}[2]{\ensuremath{\texttt{itree}\;{#1}\;{#2}}}
\newcommand{\Ret}[1]{\ensuremath{\texttt{Ret}(#1)}}
\newcommand{\Tau}[1]{\ensuremath{\texttt{Tau}(#1)}}
\newcommand{\VisA}[3]{\ensuremath{\texttt{Vis}_{#1}(#2, #3)}}
\newcommand{\Vis}[2]{\ensuremath{\texttt{Vis}(#1, #2)}}

\newcommand{\storeE}{\ensuremath{\mathsf{storeE}}}

\newcommand{\ubE}[0]{\ensuremath{\texttt{ubE}}}
\newcommand{\stepE}[0]{\ensuremath{\texttt{stepE}}}

\newcommand{\wpi}[3]{\ensuremath{ \mathrm{wpi}_{#1}\;#2\;\{ #3 \} }}

\newcommand{\Tbeh}[1]{\ensuremath{{\trg {Beh_T}}(#1)}}  
\newcommand{\Sbeh}[1]{\ensuremath{{\src {Beh_S}}(#1)}}  

\newcommand{\cmp}[1]{\ensuremath{\llbracket #1 \rrbracket}}  
\newcommand{\Corr}{\mathsf{Corr}}

\makeatletter
\renewcommand\subsection{\@startsection{subsection}{2}{\z@}%
                       {-12\p@ \@plus -4\p@ \@minus -4\p@}%
                       {6\p@ \@plus 2\p@ \@minus 2\p@}%
                       {\normalfont\normalsize\bfseries\boldmath}}
\makeatother
\makeatletter
\renewcommand\paragraph{\@startsection{paragraph}{4}{\z@}%
        {-4\p@ \@plus -2\p@ \@minus -2\p@}%
        {-0.5em \@plus -0.22em \@minus -0.1em}%
        {\normalfont\normalsize\itshape}}
\makeatother
\AtBeginDocument{%
  \setlength{\abovedisplayskip}{5pt plus 2pt minus 2pt}%
  \setlength{\belowdisplayskip}{5pt plus 2pt minus 2pt}%
  \setlength{\abovedisplayshortskip}{2pt plus 1pt}%
  \setlength{\belowdisplayshortskip}{3pt plus 1pt minus 1pt}}

\makeatletter

\spn@wtheorem{remark}{Remark}{\bfseries}{\rmfamily}
\makeatother

\setboolean{showcomments}{false}

\begin{document}

\title{Preserving Separation Logic Specifications through Compilation}

\author{Carmine Abate\inst{1} \and
  Kleio Liotati\inst{1,2} \and
  Franti\v{s}ek Farka\inst{1} \and
  Sebastian Ertel\inst{1}}
\authorrunning{C. Abate et al.}
\institute{Barkhausen Institut, Dresden, Germany\\
  \email{firstname.lastname@barkhauseninstitut.org}
  \and
  TU Dresden, Germany}

\maketitle
\begin{abstract}
Program specifications in separation logic describe ownership of resources
and serve as logical interfaces for local reasoning.
Compiler correctness, expressed as semantics preservation, does not by itself
supply such an interface for compiled code, because it does not explicitly
relate source and target resources and assertions.

We investigate WP-preserving compilation: when a compiler translates a source
program, the weakest preconditions that the program satisfies in a separation
logic carry over to the compiled program.
Our proof technique consists of a resource-aware simulation over interaction trees:
it relates the resources required by corresponding source and target events, and
tracks ownership through target operations, such as register accesses, that
have no source counterpart.
We apply the approach to a compiler from a structured imperative language to
an assembly-like language, first for closed programs and then extended with
modules and function calls.
In the extended setting, each module can be verified in the source logic and
compiled separately; its specifications carry over to the compiled module and
are preserved when compiled modules are linked.
All results are mechanized in Rocq.

\end{abstract}

\section{Introduction}\label{sec:intro}

Program logics establish specifications through which clients reason about
programs and components~\cite{VST}. In a separation logic, 
specifications describe the pre- and postconditions of a program in terms of \emph{resources}, e.g.
what part of the memory is accessible to or owned by the program and what 
resource is returned after execution~\cite{seplogicReynolds,seplogicOH}.
Such specifications serve as \emph{logical interfaces}: 
clients can use them to reason about a component without inspecting
its implementation and retain ownership of separately held resources~\cite{seplogicOH07,seplogicParkinson}.

Semantics preservation alone does not supply a corresponding logical
interface for compiled code. Correctness theorems such as those of 
CompCert~\cite{compcert} and CakeML~\cite{cakeML} relate source and target
behaviors, guaranteeing that every behavior of the compiled program is a possible behavior of 
the source program:\footnote{Throughout the paper, we use
  {\color{RoyalBlue}\sffamily blue sans-serif text}, optionally with a
  subscript $S$, for source-language entities, and
  {\color{RedOrange}\bfseries red bold text}, optionally with a subscript $T$,
  for target-language entities; see~\cite{colors} for a discussion of this
  convention. We write $\cmp{P}$ for the target program obtained by compiling
  the source program $P$.}
\[
  \Tbeh{\cmp{P}} \subseteq \Sbeh{P}.
\]
Compositional correctness and contextual refinement extend behavioral
guarantees to components
interacting with compatible contexts~\cite{compositionalCompcert,CompCertO},
while secure compilation considers arbitrary target
contexts~\cite{surveySECOMP,fac,beyondFA}.

These guarantees constrain the behaviors of compiled programs, but do not
by themselves specify how source assertions and ownership requirements
are represented in the target. A source variable may be
represented by a target memory cell, while compiled code may additionally
require registers to compute intermediate results. Deriving a corresponding
target specification requires relating resources and their ownership,
and establishing that the resulting logical interface supports reasoning
about the compiled component. The following example illustrates this gap.
\begin{example}\label{ex:running}
Consider a library component whose source-level implementation is
\[
M : \quad x \mathrel{:=} x + 1
\]
and whose small-footprint specification is the Hoare triple
\[
\{x \mathrel{\src{\mapsto}} n\}\; M\; \{x \mathrel{\src{\mapsto}} n+1\}.
\]
A client, reasoning about $M$ in separation logic, may rely on the frame
rule to obtain a richer specification,
\[
\frac{
  \{x \mathrel{\src{\mapsto}} n\}\; M\; \{x \mathrel{\src{\mapsto}} n+1\}
}
{
  \{x \mathrel{\src{\mapsto}} n ~*~ y \mathrel{\src{\mapsto}} 0\}\; M\;
  \{x \mathrel{\src{\mapsto}} n+1 ~*~ y \mathrel{\src{\mapsto}} 0\}
}
\]
that is, no variable other than $x$ is modified by $M$.

Suppose now the compiler transforms $M$ into the low-level implementation
\[
\cmp{M} : \quad
\texttt{load}~r_0,a_x;
\quad
\texttt{mov}~r_1,1;
\quad
\texttt{add}~r_0,r_0,r_1;
\quad
\texttt{store}~a_x,r_0,
\]
where $a_x$ is the memory cell holding $x$.
Semantic compiler correctness relates the behavior of $\cmp{M}$ to that
of $M$, but it gives no logical interface for $\cmp{M}$. To obtain one
without verifying the compiled code from scratch, we must relate source locations
and their ownership to target memory, account for the registers used by
$\cmp{M}$, and show that the source proof---including the fact that a
framed $y$ is left untouched---yields a target weakest precondition for
$\cmp{M}$.
\end{example}
This paper investigates \emph{WP-preserving compilation} (or simply
\emph{WP-preservation}) in a semantic, separation-logic setting:
when a compiler translates a source program, the weakest preconditions
that the program satisfies carry over to the compiled program.\footnote{Hoare triples are defined from weakest
preconditions~\cite{iris} (cf.\ Dijkstra~\cite{dijkstra75}), so preserving
weakest preconditions means preserving Hoare triples.}

\paragraph{Contributions.}
More precisely, this paper makes the following contributions:
\begin{enumerate}
\item We develop a proof technique for establishing WP-preservation based
      on a \emph{resource-aware} simulation over interaction trees (\itrees).
      \begin{enumerate}
        \item We identify an interface of three conditions under which a simulation
              transports weakest preconditions from source to target
              (the \emph{WP-transfer lemma}, \Cref{sec:abstract-transfer}),
              independently of the event signatures, return types, and
              logical handlers of the source and target \itree semantics.
        \item We build such a simulation between IMP, a structured
              imperative language, and ASM, an assembly-like language with
              registers and memory (\Cref{sec:imp-asm-simulation}).
              The simulation is defined over
              the events of the two languages rather than over a particular
              compiler, and it is closed under bind and iteration. It relates
              the resources required by corresponding events, carries their
              ownership forward to subsequent simulation steps, and accounts
              for target operations, such as register accesses, that have no
              corresponding source event.
      \end{enumerate}
\item We demonstrate our approach on the IMP-to-ASM compiler of
      Xia et~al.~\cite{itrees}. We instantiate the Program Logics \`a la
      Carte framework~\cite{plalacarte} to obtain an Iris instance for IMP,
      and prove a WP-preservation theorem for closed programs~(\Cref{sec:imp-asm-wp-pres}).
\item We extend the languages and the compiler with modules and function calls, and
      establish WP-preservation under separate compilation
      (\Cref{sec:secomp}): each module can be verified in the source logic and
compiled separately; its specifications carry over to the compiled
      module and are preserved when compiled modules are linked. We
      illustrate this with two modules whose mutually recursive functions
      decrement a shared counter, passing ownership of the counter across
      module boundaries. 
\item All results are formalized in Rocq, using the Iris framework~\cite{iris}.
\end{enumerate}

\paragraph{Structure of the paper.}
\Cref{sec:background} reviews interaction trees and their use in deriving
program logics. Readers familiar with this background may proceed directly
to~\Cref{sec:overview}, which illustrates WP-preservation and our simulation
technique through the running example. \Cref{sec:poc} presents the IMP and
ASM languages and the compiler. \Cref{sec:wp-pres} presents both layers of
our proof technique, the WP-transfer lemma and the resource-aware simulation,
and the proof of WP-preservation for the IMP-to-ASM compiler.
\Cref{sec:secomp} extends languages and compiler with modules and function
calls and establishes WP-preservation under separate compilation.
\Cref{sec:relatedW} discusses related work, and~\Cref{sec:conclusions}
concludes with current limitations and directions for future work.

\section{Background on Interaction Trees}\label{sec:background}

Interaction trees (\itrees) are a coinductive representation of recursive,
effectful computations in Rocq~\cite{itrees}.  An \itree of type
\itreeER{E}{R} is parameterized by an event signature $E$ and a return
type $R$.  It has three constructors:
\begin{itemize}
\item \Ret{r} terminates and returns $r : R$;
\item \Tau{t} takes an internal, unobservable step and continues as
      $t$; and
\item \VisA{A}{e}{k} emits an event $e : E\ A$ and, after receiving an
      answer $a : A$ from the environment, continues as $k\,a$.
\end{itemize}
Because the datatype is coinductive, its branches may be infinite.  It can
therefore represent both reactive programs that produce unbounded sequences of
events and programs that diverge silently through an infinite sequence of
$\mathsf{Tau}$ steps.

\paragraph{Equational theories.}
The main equivalence notion on \itrees is \emph{weak bisimulation} or
\emph{equivalence up to tau}, written $\approx$. It relates trees
that return related values and perform related visible events after ignoring
finite numbers of silent steps, while still distinguishing silent divergence
from termination.  More generally, the relation can be parameterized by both
a relation on return values and a relation on visible events, allowing \itrees
with different result and event types to be compared.  Other variants tailor
the observable events to a particular application; for example, equivalence
up to cutoff allows designated error events on one side to match any
subsequent behavior on the other~\cite{itreescompil}.

\paragraph{Algebraic structure.}
For every event signature $E$, \itreeER{E}{R} forms a monad whose unit is
$\mathsf{Ret}$ and whose bind $\mathsf{bind}\ t\ k$ replaces each return leaf \Ret{r} of $t$ with
the continuation $k\,r$, while preserving silent steps and visible events. 
Monadic laws hold up to weak bisimulation.

\paragraph{Computational and logical interpretations.}
Events are initially uninterpreted: an event signature describes the
interface between a computation and its environment, and a handler assigns
meaning to it. An \emph{algebraic handler}\footnote{Usually simply called
  \emph{handler}; we say \emph{algebraic} to distinguish it from the logical
  handler of Vistrup et~al.~\cite{plalacarte}.} maps each event to a
computation; an interpreter folds it over a tree, mapping $\mathsf{Ret}$ to
monadic return and commuting with bind, so that interpreters preserve the
monadic structure and the equational theory. A \emph{logical handler}
instead assigns each event a specification, the pre- and postconditions of
triggering it~\cite{plalacarte}; a weakest-precondition interpretation
lifts it to an entire \itree, mapping $\mathsf{Ret}$ to the postcondition
and composing through bind, so that an \itree semantics gives rise
compositionally to a program logic for the represented language. Both
algebraic and logical handlers for individual effects can be combined
over disjoint sums of event signatures, and both interpretations come
with laws for the iteration and recursion combinators of
\itrees~\cite{itrees,plalacarte}.

\paragraph{Program logics from \itree semantics.}
Vistrup et~al.~\cite{plalacarte} show how to derive program logics for languages with \itree semantics.
The foundation is a general weakest-precondition calculus for \itrees,
\[
\wpi{H}{t}{\Phi} \mathrel{:=} \updateMod
\begin{cases}
  \Phi(r)
    & \text{if } t = \Ret{r}, \\
      \wpi{H}{t'}{\Phi}
    & \text{if } t = \Tau{t'}, \\
  H_A\bigl(e, \lambda a.\,\wpi{H}{k(a)}{\Phi}\bigr)
    & \text{if } t = \VisA{A}{e}{k}.
\end{cases}
\]
Thus, a return establishes the postcondition, a silent step continues with the
weakest precondition of its continuation, and a visible event delegates reasoning to
the logical handler.  The outer update modality permits updates to Iris ghost
state during each step~\cite{plalacarte}.

Language-specific program logics can be obtained modularly by
providing a denotation of programs into \itreeER{E}{R} and a logical
handler for each event in $E$. 
If the language is extended with a new construct whose \itree semantics
introduces a new effect $E'$, its denotation ranges over
\itreeER{(E + E')}{R}.  The program logic can then be extended simply by
providing a logical handler for $E'$; the existing handlers and proof rules
remain unchanged.

\section{WP-preservation Overview}\label{sec:overview}

We use~\Cref{ex:running} to showcase the central ideas of our work. 
We consider a simple imperative language (IMP) with assignments, sequencing,
conditionals, and loops, compiled to an assembly-like language (ASM); their
syntax, semantics and the compiler are given in~\Cref{sec:poc}.
Recall the increment function from~\Cref{ex:running} and its specification:
\[
  M : x := x+1,
  \qquad
  \{x \mathrel{\src{\mapsto}} n\}\;M\;\{x \mathrel{\src{\mapsto}} n+1\}.
\]
The precondition supplies ownership of $x$ as well as its current value.
The postcondition returns that ownership with the updated value. A client
can therefore reason about the command without inspecting its implementation and
can retain separately owned resources, such as $y \mathrel{\src{\mapsto}} 0$ for $y \neq x$,
while executing it.

\paragraph{Establishing the specification in the source.}
The proof that $M$ satisfies this specification is a single application
of the assignment rule. Accessing $x$ requires ownership of $x$, so the
assignment consumes $x \mathrel{\src{\mapsto}} n$, evaluates $x+1$, and
returns $x \mathrel{\src{\mapsto}} n+1$.

\paragraph{Establishing the specification in the target.}
\begin{figure}[htbp]
\centering
$\displaystyle
\cmp{M} : \quad
\begin{array}{@{}l@{\quad}l@{,\;}l@{}l@{\qquad}l@{\;}c@{\;}l@{\;}c@{\;}l@{}l@{}}
  \multicolumn{4}{@{}l@{\qquad}}{}   & \{r_0 \mathrel{\trg{\mapsto}} u_0 & * & r_1 \mathrel{\trg{\mapsto}} u_1 & * & a_x \mathrel{\trg{\mapsto}} n & \}\\[.6ex]
  \texttt{load}  & r_0 & a_x &        & \{r_0 \mathrel{\trg{\mapsto}} n & * & r_1 \mathrel{\trg{\mapsto}} u_1 & * & a_x \mathrel{\trg{\mapsto}} n & \}\\
  \texttt{mov}   & r_1 & 1   &        & \{r_0 \mathrel{\trg{\mapsto}} n & * & r_1 \mathrel{\trg{\mapsto}} 1 & * & a_x \mathrel{\trg{\mapsto}} n & \}\\
  \texttt{add}   & r_0 & r_0 & ,\;r_1 & \{r_0 \mathrel{\trg{\mapsto}} n+1 & * & r_1 \mathrel{\trg{\mapsto}} 1 & * & a_x \mathrel{\trg{\mapsto}} n & \}\\
  \texttt{store} & a_x & r_0 &        & \{r_0 \mathrel{\trg{\mapsto}} n+1 & * & r_1 \mathrel{\trg{\mapsto}} 1 & * & a_x \mathrel{\trg{\mapsto}} n+1 & \}
\end{array}
$
\caption{Compiled code for $M$ with its precondition and ownership assertions
         after each instruction.}
\label{fig:increment-target}
\end{figure}
\Cref{fig:increment-target} shows $\cmp{M}$, the compiled code for $M$,
together with its precondition and the assertions that hold after each
instruction. The precondition owns the registers $r_0$ and $r_1$ used by
$\cmp{M}$ and the memory cell $a_x$ of the source variable $x$, which
contains $n$.

After \texttt{load}, register $r_0$ also contains $n$.
The \texttt{mov} instruction sets the temporary register $r_1$ to $1$.
The \texttt{add} instruction adds the values in $r_0$ and $r_1$ and stores
the result in $r_0$. The final \texttt{store} instruction writes $n+1$
to the memory cell at $a_x$.

A direct proof at the target level thus establishes:
\[
  \begin{gathered}
    \{r_0 \mathrel{\trg{\mapsto}} u_0
      * r_1 \mathrel{\trg{\mapsto}} u_1 
      * a_x \mathrel{\trg{\mapsto}} n
    \}\\
    \cmp{M}\\
    \{  r_0 \mathrel{\trg{\mapsto}} n+1
      * r_1 \mathrel{\trg{\mapsto}} 1
      * a_x \mathrel{\trg{\mapsto}} n+1
    \}.
  \end{gathered}
\]

This argument suggests that the source specification for $M$ survives compilation:
if $x$ is incremented by $1$, then $a_x$ is also incremented by $1$. Still,
the target specification also requires ownership of the registers used by
$\cmp{M}$.

\vspace{0.5em}
Our goal is to obtain a target guarantee for $\cmp{M}$ from the source proof,
without re-proving it at the target level. We achieve this by establishing a simulation
relation---over the \itrees\ of $M$ and $\cmp{M}$---
that accounts for the resources accessed by both the source and compiled programs.

\begin{figure}[htbp]
\centering
\small
\[
\begin{array}{@{}l@{\qquad}l@{}}
  \src{\textsf{Source}~t}
    & \trg{\textsf{Target}~t}\\
  \text{itree of }M
    & \text{itree of }\cmp{M}\\[.8ex]
  \hline
  \noalign{\vskip .8ex}
  \{x \mathrel{\src{\mapsto}} n\}
    & \{r_0 \mathrel{\trg{\mapsto}} u_0 * r_1 \mathrel{\trg{\mapsto}} u_1 * a_x \mathrel{\trg{\mapsto}} n\}\\[.6ex]
  v_S \leftarrow \mathsf{GetVar}(x)
    & v_T \leftarrow \mathsf{Load}(a_x)\\
  \{x \mathrel{\src{\mapsto}} n\}
    & \{r_0 \mathrel{\trg{\mapsto}} u_0 * r_1 \mathrel{\trg{\mapsto}} u_1 * a_x \mathrel{\trg{\mapsto}} n\}\\[.6ex]
  \text{---}
    & \mathsf{SetReg}(r_0,v_T)\\
  
    & \{r_0 \mathrel{\trg{\mapsto}} n * r_1 \mathrel{\trg{\mapsto}} u_1 * a_x \mathrel{\trg{\mapsto}} n\}\\[.6ex]
  \text{---}
    & \mathsf{SetReg}(r_1,1)\\
  
    & \{r_0 \mathrel{\trg{\mapsto}} n * r_1 \mathrel{\trg{\mapsto}} 1 * a_x \mathrel{\trg{\mapsto}} n\}\\[.6ex]
  \text{---}
    & p \leftarrow \mathsf{GetReg}(r_0)\\
    & \{r_0 \mathrel{\trg{\mapsto}} n * r_1 \mathrel{\trg{\mapsto}} 1 * a_x \mathrel{\trg{\mapsto}} n\}\\[.6ex]
  \text{---}
    & q \leftarrow \mathsf{GetReg}(r_1)\\
    & \{r_0 \mathrel{\trg{\mapsto}} n * r_1 \mathrel{\trg{\mapsto}} 1 * a_x \mathrel{\trg{\mapsto}} n\}\\[.6ex]
  \text{---}
    & \mathsf{SetReg}(r_0,p+q)\\
  
    & \{r_0 \mathrel{\trg{\mapsto}} n+1 * r_1 \mathrel{\trg{\mapsto}} 1 * a_x \mathrel{\trg{\mapsto}} n\}\\[.6ex]
  \text{---}
    & w \leftarrow \mathsf{GetReg}(r_0)\\
    & \{r_0 \mathrel{\trg{\mapsto}} n+1 * r_1 \mathrel{\trg{\mapsto}} 1 * a_x \mathrel{\trg{\mapsto}} n\}\\[.6ex]
  \mathsf{SetVar}(x,v_S+1)
    & \mathsf{Store}(a_x,w)\\
  \{x \mathrel{\src{\mapsto}} n+1\}
    & \{r_0 \mathrel{\trg{\mapsto}} n+1 * r_1 \mathrel{\trg{\mapsto}} 1 * a_x \mathrel{\trg{\mapsto}} n+1\}
\end{array}
\]
\caption{Resource-aware simulation between the interaction trees of $M$ and $\cmp{M}$.
         We display the postconditions after each stateful computation and omit
         silent steps, synchronized \stepE, and failure branches.
         Initially, the source variable $x$ and the target memory cell at $a_x$ both
         contain $n$, while the owned registers contain arbitrary values $u_0$ and $u_1$.
         For the simulation to hold, $\mathsf{GetVar}(x)$ must match
         $\mathsf{Load}(a_x)$, and $\mathsf{SetVar}(x,\cdot)$ must match
         $\mathsf{Store}(a_x,\cdot)$. Each match must also
         provide the ownership of the resources required by the next proof obligation.
         The target's $\mathsf{GetReg}$ and $\mathsf{SetReg}$ events have no
         corresponding source events; the source stutters, as indicated by
         the dashes, while register ownership is preserved.
         The simulation owns only the cell $a_x$, the source authority of
         $x$, and the registers of $\cmp{M}$.
        }
\label{fig:increment-itrees}
\end{figure}

\paragraph{Resource-aware simulation.}
\Cref{fig:increment-itrees} shows a \emph{resource-aware} simulation between
${\src t}$ and ${\trg t}$, the \itrees\ of $M$ and $\cmp{M}$, respectively.
For the simulation to hold, the read $\mathsf{GetVar}(x)$ must match the
read $\mathsf{Load}(a_x)$, and the match must relate the ownership of the source
and target resources.

After the initial synchronized reads, the simulation cannot proceed in lock-step.
The target writes the loaded value to $r_0$, with no corresponding source event.
Further register reads and writes compute the increment and prepare the value for
the final store. The simulation lets the source stutter, as indicated by the dashes,
while the target performs the necessary register accesses. The simulation tracks 
register contents and carries ownership of resources forward to subsequent proof obligations.
A $\mathsf{GetReg}$ preserves both the register's contents and its ownership.
A $\mathsf{SetReg}$ updates the contents while retaining ownership.

\paragraph{Transporting the specification.}
A source assertion such as $x\mathrel{\src{\mapsto}}v$ speaks about the
source store, which $\cmp{M}$ never touches. The simulation therefore
maintains a \emph{state correspondence} $\Corr_M$ for the variables of $M$
and the registers of $\cmp{M}$: for $x$ it owns the target cell $a_x$ and the
source \emph{authority} of $x$, a ghost resource that agrees with the
client's $x\mathrel{\src{\mapsto}}v$, with the same entry in both, and it
owns $r_0$ and $r_1$. $\Corr_M$ speaks about $x$, $a_x$, $r_0$ and $r_1$
only, not about the rest of the store or the memory.
Our preservation theorem returns the source postcondition together with
$\Corr_M$. From $x\mathrel{\src{\mapsto}}n{+}1$ and the authority of $x$ in
$\Corr_M$, the cell $a_x$ holds $n{+}1$, so we extract
$a_x\mathrel{\trg{\mapsto}}n{+}1$. Since the target precondition owns no
other cell, the frame rule of the target logic adds any other cell, such as
$a_y\mathrel{\trg{\mapsto}}0$. Combining the source proof of $M$ with our
theorem, we obtain
\[
  \Corr_M * x\mathrel{\src{\mapsto}}n * a_y\mathrel{\trg{\mapsto}}0
  \;\vdash\;
  \wpi{\trgH}{\denoteASM{\cmp{M}}(0)}{a_x\mathrel{\trg{\mapsto}}n{+}1
     * a_y\mathrel{\trg{\mapsto}}0 * {\trg\theta}}
\]
where ${\trg\theta}$ owns the registers of $\cmp{M}$ (\Cref{sec:imp-asm-wp-pres}).

\smallskip
A relation between source and target states induces a mapping from source
to target properties~\cite{tracerel}.
$\Corr_M$ plays this role, mapping the source postcondition
$x\mathrel{\src{\mapsto}}n{+}1$ to $\Corr_M * x\mathrel{\src{\mapsto}}n{+}1$,
which in turn entails the target postcondition
$a_x\mathrel{\trg{\mapsto}}n{+}1 * {\trg\theta}$
because $\Corr_M$ owns the cell $a_x$ and an authority that agrees with the
source points-to assertion.\footnote{This works for points-to assertions;
an arbitrary source assertion need not have a target counterpart.}

\smallskip
The simulation relation also requires ownership of $r_0$ and $r_1$.
The value eventually written to $a_x$ is computed 
through these registers. The \texttt{load} places $n$ in $r_0$, \texttt{mov} places
$1$ in $r_1$, and \texttt{add} computes their sum before \texttt{store} writes it to memory.
Proving the target postcondition therefore requires permission to read and
update these registers, together with assertions tracking their contents.
Their initial values may be arbitrary, since both registers are overwritten
before being read, but their ownership must be available from the start.
The simulation carries this additional ownership through the intermediate
steps and returns it together with ownership of the updated cell at $a_x$.

\section{IMP-to-ASM: Languages, Logic, and Compilation}\label{sec:poc}
We use the IMP-to-ASM compiler of Xia et~al.~\cite{itrees}
to instantiate our approach. This section presents the two languages, their
\itree semantics, the source program logic, and the compilation scheme.
These ingredients provide the setting for WP-preservation in~\Cref{sec:wp-pres}.

\begin{figure}[htbp]
\centering
\small
$
\begin{array}{@{}r@{\;}c@{\;}l@{\qquad}r@{\;}c@{\;}l@{}}
\multicolumn{3}{@{}l}{\textbf{IMP}} & \multicolumn{3}{@{}l}{\textbf{ASM}}\\[-.2ex]
c & ::= & \mathsf{skip}\mid x:=a\mid c;c
  & i & ::= & \mathsf{mov}\ r,o\mid\mathsf{op}\ r,r,o\mid\mathsf{not}\ r,o
          \mid\mathsf{load}\ r,\alpha\mid\mathsf{store}\ \alpha,o\\
  & \mid & \mathsf{if}\ b\ \mathsf{then}\ c\ \mathsf{else}\ c
  & j & ::= & \mathsf{jmp}\ \ell\mid\mathsf{brz}\ r,\ell_0,\ell_1\mid\mathsf{halt}\\
  & \mid & \mathsf{while}\ b\ \mathsf{do}\ c
  & B & ::= & j\mid i;B
\end{array}
$
\caption{IMP and ASM syntax. In IMP, $n\in\mathbb{N}$ and
$x$ are variables identifiers, $a$ is an arithmetic expressions and $b$ a Boolean expression.
In ASM, registers $r\in\mathbb{N}$, addresses $\alpha\in\mathit{Addr}=\mathit{Ident}$,
and labels $\ell$ range over a finite set $L$. An operand $o$ is an immediate
value $n$ or a register $r$, and $\mathsf{op}$ ranges over addition,
subtraction, multiplication, equality, order, and conjunction. A block $B$ ends
in a branch.}
\label{fig:imp-asm-syntax}
\end{figure}

\subsection{Source Language: IMP}\label{sec:poc-imp}
IMP programs are simple sequential commands (\Cref{fig:imp-asm-syntax}). 
Its state is a finite partial map from variable identifiers to optional
natural numbers, $\sigma:\mathit{Ident}\rightharpoonup\mathsf{option}(\mathbb{N})$, called the \emph{store}.
Absent entries denote undeclared variables, and $\mathsf{None}$ entries
denote unset variables.
The denotation of a command $c$, $\denoteIMP{c}: \itreeER{\srcE}{\mathsf{unit}}$, is an \itree\ 
returning $\mathsf{unit}$, with event signature $\srcE=\storeE+\stepE+\ubE$, 
where $\storeE$ consists of the events $\mathsf{GetVar}(x)$, which reads
the entry of the variable $x$, and $\mathsf{SetVar}(x,w)$, which writes $w$
to it.
Reading an uninitialized variable, or accessing an undeclared one, triggers
an undefined behavior, $\ubE$.
Assignments evaluate their expressions and write to the store; sequencing uses bind,
conditionals select a continuation, and loops use \itree iteration.
\stepE{} is an administrative event with no operational effect; it records
that a computational step has occurred.\footnote{In~\Cref{sec:imp-program-logic} we use \stepE{} to reason under the later modality.}

\subsection{IMP Program Logic}\label{sec:imp-program-logic}
We use events and logical handlers provided by the mechanized
development accompanying \emph{Program Logics \`a la Carte}~\cite{plalacarte}
to derive a weakest precondition calculus for IMP from its \itree semantics
(full details in~\Cref{app:imp-logic}).

\smallskip
For a command $c$ and postcondition ${\src\Phi}$, its semantic weakest precondition is
\[
 \wpIMP{c}{{\src\Phi}}\;:=\;
 \wpi{\srcH}{\denoteIMP{c}}{{\src\Phi}},
 \qquad
 \srcH=\handler{store}\oplus\handler{step}\oplus\handler{ub}.
\]
To ease the presentation, hereafter we suppress masks.

\handler{store} is the logical handler for \storeE. It specifies each event
by points-to assertions on the accessed variable only. We write
$x\mathrel{\src{\mapsto^{?}_{q}}}w$ for ownership of a fraction $q$ of $x$,
whose entry is $w:\mathsf{option}(\mathbb N)$; we omit $q$ for full
ownership, and $x\mathrel{\src{\mapsto}}v$ abbreviates
$x\mathrel{\src{\mapsto^{?}}}\mathsf{Some}(v)$. Then
\[
\begin{aligned}
 \handler{store}(\mathsf{GetVar}(x),{\src\Psi})&:=
   \exists q,w.\;x\mathrel{\src{\mapsto^{?}_{q}}}w*
   (x\mathrel{\src{\mapsto^{?}_{q}}}w\wand{\src\Psi}(w)),\\
 \handler{store}(\mathsf{SetVar}(x,w'),{\src\Psi})&:=
   \exists w.\;x\mathrel{\src{\mapsto^{?}}}w*
   (x\mathrel{\src{\mapsto^{?}}}w'\wand{\src\Psi}(\mathsf{tt})).
\end{aligned}
\]
A read may use any fraction and returns it unchanged; a write requires full
ownership. Since reading an unset variable is undefined behavior, reading
a variable in a proof requires an initialized value,
$x\mathrel{\src{\mapsto}}v$. The handler mentions no resource over the whole
store. The points-to assertion is a fragment of a ghost authority kept for
each variable separately (\Cref{app:imp-logic}), so the fragment and the
authority of a variable can be held by different parties; the simulation
of~\Cref{sec:wp-pres} relies on this. The authorities of all variables of
the store appear only in the adequacy proof of the handler, which relates
it to the computational semantics of the store.

The undefined-behavior handler assigns the false assertion to an \ubE{} event,
independently of its continuation, $ \handler{ub}(\ubE,{\src\Psi}) := \bot.$
Consequently, proving the WP of a program requires showing that undefined
behavior is unreachable.

The step handler,
$\handler{step}(\stepE,{\src\Psi}) := \triangleright_m\,{\src\Psi}(\mathsf{tt})$,
places a modality $\triangleright_m$ around the continuation, where the
parameter $m$ selects either the later modality,
$\triangleright_{\mathsf{Later}}P=\triangleright P$, or no modality,
$\triangleright_{\mathsf{Identity}}P=P$.
With $m=\mathsf{Later}$, each \stepE{} records one unit of logical progress
and permits guarded recursive reasoning about the remainder of the
computation.

The proof system provides rules for each IMP command. The rule for sequencing
follows semantic bind, conditionals select the rule for the branch determined
by the guard, and the loop rule unfolds the loop once with guarded recursive
obligations. The rules of consequence and framing hold as WP lemmas, so a
proof may retain resources untouched by a command.
\Cref{app:imp-logic} gives the expression and command rules, their
structural rules, and the precise command WP used in the development.

An inductive derivability judgment packages these rules; we write
$\{P\}\;c\;\{{\src\Phi}\}$ when the triple is derivable. It includes consequence
but not framing, and its loop rule works on the once-unrolled loop
instead of the guarded one (\Cref{app:imp-logic}).
Its soundness theorem states that every derivable triple
$\{P\}\;c\;\{{\src\Phi}\}$ entails \mbox{$P\vdash\wpIMP{c}{{\src\Phi}}$}.

\subsection{Target Language: ASM}\label{sec:poc-asm}
ASM is a small assembly-like language organized as control-flow graphs
(\Cref{fig:imp-asm-syntax}).
Its state consists of registers and memory, both finite partial maps to
optional natural numbers, so the semantics distinguishes an absent location from a location that exists
but is uninitialized.

\vspace{0.5em}
The nodes of an ASM control-flow graph are basic blocks $B$. A block is a finite
sequence of instructions $i$ followed by exactly one terminator $j$.
Instructions perform straight-line computation, moving values between
registers, applying arithmetic or Boolean operations, and loading from or
storing to memory. The terminator determines what happens after the instruction sequence.
An unconditional jump selects one successor, a zero-test branch selects one of
two successors from the value of a register, and \textsf{halt} stops execution.
Thus instructions transform the machine state within a block, whereas the
terminator controls execution between blocks.

Formally, an $\mathsf{asm}(a,b)$ graph has $a$ visible entry labels, $b$
visible exit labels, and some finite number $I$ of internal labels. Its code
maps each element of $\mathsf{fin}(I+a)$---an internal label or a visible
entry---to a block. A non-halting block transfers control to an element of
$\mathsf{fin}(I+b)$: an internal label continues execution at another block,
while a visible exit returns control to the surrounding graph. Here
$\mathsf{fin}(k)$ denotes the finite type with $k$ elements. This distinction
between visible and internal labels allows graphs to be linked and internal
control-flow edges to be hidden compositionally.

The event signature for ASM is 
$\trgE=\mathsf{regE}+\mathsf{memoryE}+\mathsf{Exit}+\stepE$.
$\mathsf{regE}$ consists of $\mathsf{GetReg}(r)$ and $\mathsf{SetReg}(r,w)$,
and $\mathsf{memoryE}$ of $\mathsf{Load}(\alpha)$ and $\mathsf{Store}(\alpha,w)$;
each event accesses a single register or memory cell. Reads of unset
locations and \textsf{halt} emit $\mathsf{Exit}$.
The denotation of a block executes its instructions in order and then follows
its terminator.
The semantics of a control-flow graph iterates this process across internal
labels until it reaches a visible exit, giving
$\denoteASM{p}:\mathsf{fin}(a)\to\itreeER{\trgE}{\mathsf{fin}(b)}$.

\subsection{Compilation}\label{sec:poc-compiler}

The compilation function $\mathsf{compile}:\mathsf{com}\to\mathsf{asm}(1,1)$
is the one of Xia et~al.~\cite{itrees}, used unchanged for closed programs
and extended with function calls in~\Cref{sec:secomp}; the \itree
semantics of both languages faithfully follow Xia et~al., but access the
state through one event per variable, register or memory cell, and
introduce \stepE{} and \ubE. The
compiler follows the source syntax. Expression compilation emits instruction lists: constants
become moves and variables become loads. Binary operands are computed into
$r$ and $r+1$ and combined into $r$, using higher registers as scratch space
while preserving lower registers. Boolean expressions follow the same scheme.
Assignment computes into register zero and stores the result at the address
$a_x$ of the assigned variable $x$, with $a_x:=x$.

Graph combinators implement control flow. Skip is the identity graph and
assignment a single block. Sequencing links one graph's exit to the next
entry; conditionals test register zero and merge branch exits; loops link
the body back to the guard. Relabelling and hiding internal labels implement
these connections. Their denotation laws express sequencing as bind and
looping as \itree iteration, supporting compositional compiler proofs.

\section{WP-Preservation}\label{sec:wp-pres}

We now show that the compiler of~\Cref{sec:poc} preserves source-level
weakest preconditions, and hence source Hoare triples.
The proof is organized around a generic transfer principle for weakest
preconditions over interaction trees (\Cref{sec:abstract-transfer}).
This result is independent of IMP, ASM, and the compiler, as it is
parameterized by a \emph{resource-aware} simulation over \itrees\ with
arbitrary event signatures and logical handlers.
We then instantiate this simulation for IMP and ASM
(\Cref{sec:imp-asm-simulation}) and use its composition properties to
prove WP-preservation for the IMP-to-ASM compiler
(\Cref{sec:imp-asm-wp-pres}). The full definitions and technical
statements appear in~\Cref{app:wp-pres}.

\subsection{Abstract Transfer Lemma}\label{sec:abstract-transfer}

To transport a source WP to a target WP, we must relate both the
computations and the resources mentioned in their specifications.
We therefore use an Iris-valued simulation
\[
 \SIM_R:
 \itreeER{\srcE}{\src A}\to
 \itreeER{\trgE}{\trg A}\to\iProp\;\Sigma,
\]
indexed by a relation over returned values $R:{\src A}\to{\trg A}\to\iProp\;\Sigma.$
The logical handlers and the simulation use the same Iris resource
universe $\Sigma$, containing both source and target resources.
The assertion $R(u,v)$ relates the returned values and describes the
resources available when both \itrees\ return a value.

\smallskip
Hereafter we provide an intuitive description of the requirements we impose
on $\SIM_R$; their rigorous statements are in~\Cref{app:abstract-transfer}.

\paragraph{Returns (\textup{\textsc{SIMret}}).}
$\SIM_R$ must relate a source $\Ret{u}$ with a target \itree\ ${\trg{t}}$ that triggers
visible events before returning $v$ with $R(u,v)$,
as long as the target logical handler does not change the target postcondition.
For example, if the target postcondition concerns only the target memory,
the simulation can relate a source return to a register event, which does
not affect a predicate over the memory, followed by a target return
\[
 \SIM_R\bigl(\Ret{\srcRetv{\mathsf{tt}}},
   \Vis{\mathsf{SetReg}(r_0,v)}
        {\lambda \_.\,\Ret{\trgRetv{\mathsf{tt}}}}\bigr).
\]

\paragraph{Tau steps (\textup{\textsc{SIMtau}}).}
A silent source step exposes no effect. The simulation must allow us to
remove it and continue relating the remaining source computation to the
same target computation.

\paragraph{Visible events (\textup{\textsc{SIMvis}}).}
A visible source event must be matched by a target event, possibly after
finitely many additional target steps, such as the register write above.
The simulation must justify these steps through the target logical handler.

For the matched events, an event relation identifies which source and
target effects correspond. A resource-aware answer relation describes
how their answers are related and which resources are available to
simulate their continuations. Compatibility of the logical handlers
allows the source handler's proof obligations to establish those of the
target handler. Together, these conditions let us extend the source WP
proof through the matched event while accounting for the target's
additional computation.

\smallskip
These local requirements suffice to transport an entire WP.

\begin{lemma}[Abstract WP transfer]\label{lem:abstract-wp-transfer}
Suppose $\SIM_R$ satisfies \textup{\textsc{SIMret}}, \textup{\textsc{SIMtau}}, and
\textup{\textsc{SIMvis}} from~\Cref{app:abstract-transfer}.
The following entailment holds:
\[
\begin{aligned}
 &\SIM_R({\src s},{\trg t})
 *\wpi{\srcH}{{\src s}}{{\src\Phi}}\\
 &*\bigl(\forall u,v.\;
       R(u,v)*{\src\Phi}(u)\wand{\trg\Psi}(v)\bigr)
 \;\vdash\;\wpi{\trgH}{{\trg t}}{{\trg\Psi}}.
\end{aligned}
\]
\end{lemma}
As before, we suppress invariant masks. Source and target weakest
preconditions are assertions of the same Iris logic and share one set of
invariants, whose contents may refer to both source and target resources.
The mechanized lemma holds for any invariant mask used by both, and more
generally lets the target open more invariants than the source, with a
fancy update from the source to the target mask at return (\Cref{lem:abstract-wp-transfer-masks}).

\subsection{A Resource-Aware Simulation for the \itrees of IMP and ASM}
\label{sec:imp-asm-simulation}

We define an instance of the simulation for the event signatures of IMP and ASM,
$\srcE=\storeE+\stepE+\ubE$ and
$\trgE=\mathsf{regE}+\mathsf{memoryE}+\mathsf{Exit}+\stepE$.
The assertion $\simIMPASM_R({\src s},{\trg t})$ relates the computations
while tracking the resources needed to match their effects.

\paragraph{Target weakest preconditions.}
We give ASM no proof rules of its own, since verification is meant to
happen at the source level, where stronger abstractions are available. By
specifying a logical interpretation of the target events $\trgE$, we obtain
from the generic construction of~\Cref{sec:background}~\cite{plalacarte} a
weakest precondition calculus for the \itrees with event signature $\trgE$,
and hence, indirectly, for ASM programs. Its meaning
is fixed by the logical handler, and no ASM-specific proof rules are needed,
since target weakest preconditions arise only as conclusions of
WP-preservation.

\smallskip
For an ASM graph $p$, entry label $\ell$, and return postcondition ${\trg\Psi}$,
the target $\wpi{\trgH}{\denoteASM{p}(\ell)}{{\trg\Psi}}$ is defined by the logical handler
\[
 \trgH:=\handler{reg}\oplus\handler{memory}
       \oplus\handler{exit}({\trg\Phi}_{\mathsf{exit}})\oplus\handler{step}.
\]
The register and memory handlers, $\handler{reg}$ and $\handler{memory}$,
follow the IMP store handler (\Cref{sec:imp-program-logic}): they are
specified by points-to assertions on the accessed location only.
$\mathsf{GetReg}(r)$ takes $r\mathrel{\trg{\mapsto^{?}_{q}}}w$ for some
fraction $q$, returns $w$ and gives the assertion back, and
$\mathsf{SetReg}(r,w')$ takes $r\mathrel{\trg{\mapsto^{?}}}w$ and gives back
$r\mathrel{\trg{\mapsto^{?}}}w'$. $\mathsf{Load}(\alpha)$ and
$\mathsf{Store}(\alpha,w')$ do the same with
$\alpha\mathrel{\trg{\mapsto^{?}}}w$.
The step handler, $\handler{step},$ places the same later modality as in
IMP around the continuation, recording logical progress.
Finally, $\mathsf{Exit}$ has no answer and hence no continuation, so
its handler requires a chosen assertion ${\trg\Phi}_{\mathsf{exit}}$, independently of ${\trg\Psi}$.
The postcondition ${\trg\Psi}$ applies when the graph returns through a visible
exit label. As for the source WP, we suppress the invariant mask and the
step modality parameter.

\paragraph{State correspondence.}
The simulation relates each source variable $x$ to the target cell $a_x$,
one variable at a time. For a set of variables $V$ it owns
\[
 \Corr_V:=\mathop{\scalebox{1.5}{$\ast$}}_{x\in V}\mathsf{Cell}(x),
 \qquad
 \mathsf{Cell}(x):=\exists w.\;\mathsf{Auth}_x(w)*a_x\mathrel{\trg{\mapsto^{?}}}w,
\]
where $\mathsf{Auth}_x(w)$ is the ghost authority of the source variable $x$
with entry $w$. The source points-to assertions remain with the source
specification. Since the authority agrees with every points-to assertion
$x\mathrel{\src{\mapsto^{?}_{q}}}w'$ for $x$, that is, $w=w'$, the
simulation knows the value of $x$ while the client owns $x$; and since it
also owns $a_x$ with the same entry, the source and target values agree.
$\Corr_V$ speaks only about the variables in $V$ and their cells, not about
the rest of the store or the memory.

The simulation also owns the temporary registers used by the compiled
command, $\mathsf{RegOwn}(X):=\mathop{\scalebox{1.5}{$\ast$}}_{r\in X}
\exists w.\;r\mathrel{\trg{\mapsto^{?}}}w$. Their contents may change during
evaluation, but their ownership is returned so subsequent computations can
reuse them.

We write $\Corr_{X,V}:=\mathsf{RegOwn}(X)*\Corr_V$
for these simulation resources.

\paragraph{Progress and administrative steps.}
The simulation must accommodate the target's extra register operations
without allowing it to perform administrative work forever.
We express this distinction using a mixed fixed point:
\begin{definition}[IMP--ASM simulation]\label{def:imp-asm-simulation}
For an Iris-valued return relation $R$,
\[
 \simIMPASM_R:=\nu G.\,\mu L.\,\simFunctor(G,L).
\]
The functor's clauses are given in~\Cref{fig:imp-asm-simulation}.
\end{definition}
The greatest fixed point $\nu$ permits an unbounded number of matched
visible events. The least fixed point $\mu$ bounds each intervening
stretch of silent steps and target register operations.
The clauses have three roles.

\emph{Matching effects.}
A source read $\mathsf{GetVar}(x)$ matches the target read
$\mathsf{Load}(a_x)$, and a source write $\mathsf{SetVar}(x,w)$ matches the
target write $\mathsf{Store}(a_x,w)$. For the matched pair, the simulation
lends $\mathsf{Cell}(x)$. For a read, the points-to assertion that the
source handler requires agrees with $\mathsf{Auth}_x$, so the target cell
holds the same entry and the two answers are equal. For a write, the full
points-to assertion and $\mathsf{Auth}_x$ are updated together, and so is
the target cell. Logical steps match on both sides; source undefined
behavior matches target exit. The resources obtained after a matched event,
here $\mathsf{Cell}(x)$, are passed to the continuation simulation
through $G$.

\emph{Allowing administrative work.}
The target may read or write registers while the source remains
unchanged. A register read requires ownership of the register and returns
it; a write requires full ownership and returns the register with its new
value. In both cases, the remaining computation is related through $L$,
so register bookkeeping cannot continue indefinitely between matched
events. Silent steps on either or both sides also recurse through $L$.

\emph{Finishing the simulation.}
When both computations return, their values and remaining resources must
satisfy $R$. For the compiler proof, we take this return relation to
reestablish $\Corr_{X,V}$, allowing the next command to start with the
same correspondence and reuse the registers.

\begin{figure}[tbp]
\centering
\footnotesize
\setlength{\abovedisplayskip}{5pt}
\setlength{\belowdisplayskip}{5pt}
\setlength{\abovedisplayshortskip}{3pt}
\setlength{\belowdisplayshortskip}{3pt}
\textbf{(a) The simulation functor.}
\[
\begin{aligned}
 \simFunctor(G,L)(s,t):={}&
 \bigl(\textit{Return:}\quad \exists u,v.\;\ulcorner s=\Ret{u}\urcorner*
                    \ulcorner t=\Ret{v}\urcorner*R(u,v)\bigr)\\
 {}\lor{}&\bigl(\textit{Source silent step:}\quad \exists s'.\;\ulcorner s=\Tau{s'}\urcorner*L(s',t)\bigr)\\
 {}\lor{}&\bigl(\textit{Target silent step:}\quad \exists t'.\;\ulcorner t=\Tau{t'}\urcorner*L(s,t')\bigr)\\
 {}\lor{}&\bigl(\textit{Register read:}\quad \exists r,k.\;\ulcorner t=\Vis{\mathsf{GetReg}(r)}{k}\urcorner\\
 &\qquad*\exists q,w.\;r\mathrel{\trg{\mapsto^{?}_{q}}}w
   *(r\mathrel{\trg{\mapsto^{?}_{q}}}w\wand L(s,k(w)))\bigr)\\
 {}\lor{}&\bigl(\textit{Register write:}\quad \exists r,w',k.\;
                 \ulcorner t=\Vis{\mathsf{SetReg}(r,w')}{k}\urcorner\\
 &\qquad*\exists w.\;r\mathrel{\trg{\mapsto^{?}}}w
   *(r\mathrel{\trg{\mapsto^{?}}}w'\wand L(s,k(\mathsf{tt})))\bigr)\\
 {}\lor{}&\bigl(\textit{Joint silent step:}\quad \exists s',t'.\;\ulcorner s=\Tau{s'}\urcorner*
                 \ulcorner t=\Tau{t'}\urcorner*L(s',t')\bigr)\\
 {}\lor{}&\bigl(\textit{Matched event:}\quad \exists A,B,e,f,k,h.\;
                 \ulcorner s=\Vis{e}{k}\urcorner*
                 \ulcorner t=\Vis{f}{h}\urcorner\\
 &\qquad* \mathcal{E} (e,f)*\mathcal{T}(e,f)
       *(\forall a,b.\;\mathcal{A}(e,a,f,b)\wand G(k(a),h(b)))\bigr).
\end{aligned}
\]

\medskip
\footnotesize
\setlength{\abovedisplayskip}{4pt}
\setlength{\belowdisplayskip}{4pt}
\setlength{\abovedisplayshortskip}{2pt}
\setlength{\belowdisplayshortskip}{2pt}
\textbf{(b) Event relations.}
\[
 \mathsf{Cell}(x):=\exists w.\;\mathsf{Auth}_x(w)*a_x\mathrel{\trg{\mapsto^{?}}}w,
 \qquad
 \Corr_V:=\mathop{\scalebox{1.5}{$\ast$}}_{x\in V}\mathsf{Cell}(x).
\]
\setlength{\tabcolsep}{4pt}
\begin{tabular}{@{}lccc@{}}
\hline
$(e,f)$ & $\mathcal{E}(e,f)$ & $\mathcal{T}(e,f)$ & $\mathcal{A}(e,a,f,b)$ \\
\hline
$(\mathsf{GetVar}(x),\mathsf{Load}(\alpha))$
 & $\ulcorner\alpha=a_x\urcorner$
 & $\mathsf{Cell}(x)$
 & $\ulcorner a=b\urcorner*\mathsf{Cell}(x)$ \\[2pt]
$(\mathsf{SetVar}(x,w),\mathsf{Store}(\alpha,w'))$
 & $\ulcorner\alpha=a_x\land w=w'\urcorner$
 & $\mathsf{Cell}(x)$
 & $\mathsf{Cell}(x)$ \\[2pt]
$(\mathsf{Step}_{\mathsf{src}},\mathsf{Step}_{\mathsf{trg}})$
 & $\mathsf{True}$ & $\mathsf{True}$ & $\mathsf{True}$ \\[2pt]
$(\mathsf{UB},\mathsf{Exit})$
 & $\mathsf{True}$ & $\mathsf{True}$ & \textit{no answers} \\[2pt]
All other pairs
 & $\mathsf{False}$ & $\mathsf{True}$ & $\mathsf{False}$ \\
\hline
\end{tabular}

\smallskip
The two $\mathsf{Step}$ events record source and target logical progress.
The UB--Exit pair has empty answer types.

\caption{The IMP--ASM simulation: (a) its single-step functor and
(b) the event relations for matched events.}
\label{fig:imp-asm-simulation}
\medskip
\begin{minipage}{\linewidth}
\footnotesize
\textbf{Notation and resource updates.}
In~\Cref{fig:imp-asm-simulation}, $\mathsf{Auth}_x(w)$ owns the full ghost
authority of the source variable $x$, recording the entry $w$. It agrees
with every source points-to assertion for $x$: from
$\mathsf{Auth}_x(w)*x\mathrel{\src{\mapsto^{?}_{q}}}w'$ we obtain $w=w'$.
Updating the entry requires both $\mathsf{Auth}_x$ and the full points-to
assertion; the matched write updates them together with the target cell.
Blue/red points-to arrows denote source/target ownership.
$\alpha\mathrel{\trg{\mapsto^{?}_{q}}}w$ owns a fraction $q$ of $\alpha$ with
an optional value $w$, the fraction is omitted when full, and
$\alpha\mathrel{\trg{\mapsto}}v$ abbreviates
$\alpha\mathrel{\trg{\mapsto^{?}}}\mathsf{Some}(v)$; likewise for source
assertions. Target memory and register ownership are distinguished by
the address or register on the left. The authorities over the whole store,
memory and register file appear only in the adequacy proofs of the logical
handlers, never in the simulation.
\end{minipage}
\end{figure}

\begin{lemma}[Transfer and algebraic properties]\label{lem:sim-algebra}
The relation $\simIMPASM_R$ satisfies \textup{\textsc{SIMret}},
\textup{\textsc{SIMtau}}, and \textup{\textsc{SIMvis}}, as in the hypotheses
of~\Cref{lem:abstract-wp-transfer}. Moreover, it
is compatible with weak bisimilarity (eutt), and is closed under bind
and iteration (assuming the source loop cannot silently restart forever).
\end{lemma}
The iteration assumption permits divergent loops that keep performing
visible events. It excludes infinitely many restarts with no such
progress. The precise well-foundedness condition and composition
lemmas appear in~\Cref{app:sim-algebra}.

\subsection{WP-Preservation for IMP-to-ASM}\label{sec:imp-asm-wp-pres}

We now use the simulation from~\Cref{def:imp-asm-simulation} to transport
IMP specifications to compiled programs. The central distinction is between
retaining an arbitrary source postcondition and obtaining assertions about target memory.

\paragraph{Transferring specifications.}
A source assertion such as $x\mathrel{\src{\mapsto}}v$ talks about the
source store, which the compiled program never touches; on its own, it says
nothing about target memory.
The preservation theorem therefore returns the source postcondition ${\src\Phi}$
together with what is needed to read it against target memory: for each
variable of the command, its source authority and its target cell, with the
same entry. From the target's point of view the authorities are ghost, since
no target instruction reads or writes them.
For a command $c$, let $V_c$ be its variables and $X_c$ the registers used
by $\mathsf{compile}(c)$. We write
$\Corr_c:=\Corr_{X_c,V_c}=\mathsf{RegOwn}(X_c)*\Corr_{V_c}$ for the
\emph{state correspondence} of $c$. It owns only the cells of the variables
of $c$ and the registers of its compiled code; no resource in it ranges
over the whole store or the whole memory.

\begin{theorem}[Closed-program WP preservation]
\label{thm:compiler-wp-pres}
For every IMP command $c$ and source postcondition ${\src\Phi}$,
\[
 \Corr_c*\wpIMP{c}{{\src\Phi}}
 \;\vdash\;
 \wpi{\trgH}{\denoteASM{\mathsf{compile}(c)}(0)}{\compProp{{\src\Phi}}_c},
\]
where $\compProp{{\src\Phi}}_c:=\Corr_c*{\src\Phi}(\mathsf{tt})$.
Source and target use the same invariant mask and step modality, and the
target exit postcondition is arbitrary.
\end{theorem}
The state correspondence is thus returned, with the final entries of the
variables of $c$, together with the source postcondition.
Since the precondition owns nothing outside the cells of $V_c$ and the
registers $X_c$, the frame rule of the target logic applies to compiled
code: any target assertion owned separately, such as
$a_z\mathrel{\trg{\mapsto}}42$ for a variable $z$ that $c$ does not use, is
preserved by $\mathsf{compile}(c)$.

\begin{proof}[sketch]
The main step is the compiler simulation. Take
$R_c(u,v):=\Corr_c$. We establish
\[
 \Corr_c\;\vdash\;
 \simIMPASM_{R_c}
   (\denoteIMP{c},\denoteASM{\mathsf{compile}(c)}(0)).
\tag{Compiler simulation}\label{eq:compiler-simulation}
\]
The proof proceeds by induction on $c$, using~\Cref{lem:sim-algebra},
for the more general correspondence $\Corr_{X,V}$ with $X\supseteq X_c$ and
$V\supseteq V_c$, so that a subcommand runs with the correspondence of the
enclosing command. Skip restores the input resources directly. For
assignment, expression simulation accounts for temporary register use, and
the matched write updates the source authority and the target cell of the
assigned variable together. For a conditional, expression
simulation relates the source guard to the target branch test, after
which we use the selected branch's induction hypothesis.
Sequencing uses closure under bind to pass resources between commands.
The while case uses closure under iteration, relating source loop states
to target control-flow labels. The logical step following guard
evaluation rules out an infinite chain of silent restarts.

Applying~\Cref{lem:abstract-wp-transfer} to this simulation and the source
WP yields a target WP that returns the source postcondition together with
$\Corr_c$, that is, $\compProp{{\src\Phi}}_c$.
Soundness of the source proof system supplies the source WP premise
for derivable Hoare triples.\qed
\end{proof}

\paragraph{Obtaining target points-to assertions.}
In practice, a client of the compiled program wants target assertions rather
than a source postcondition read through the state correspondence.
An arbitrary Iris assertion need not admit such a translation, which is why
\Cref{thm:compiler-wp-pres} retains ${\src\Phi}$ together with the state
correspondence. Points-to assertions, however, do.
If the source postcondition owns a variable $x$ of $c$ containing $v$, the
authority of $x$ in $\Corr_c$ agrees with it, so the cell $a_x$ owned by
$\Corr_c$ holds $v$:
$x\mathrel{\src{\mapsto}}v*\mathsf{Cell}(x)\vdash a_x\mathrel{\trg{\mapsto}}v$.
For example, $x\mathrel{\src{\mapsto}}v*y\mathrel{\src{\mapsto}}w$ lets us
extract $a_x\mathrel{\trg{\mapsto}}v*a_y\mathrel{\trg{\mapsto}}w$, where $a_x$
and $a_y$ are the memory addresses of $x$ and $y$; the cells of the other
variables of $c$ are returned in $\Corr_c$.
This extraction uses the supplied target ownership; the source assertion
alone does not create target resources.
\Cref{lem:pointsto-transfer} states it precisely, and
\Cref{thm:compiler-wp-pres-frame} combines it with~\Cref{thm:compiler-wp-pres}.
Variables that $c$ does not use need no such detour: their target cells are
framed in the target logic, as above.

For a closed program, the target facts obtained this way, points-to
assertions about the final memory, could also be derived from the adequacy
of the source logic and a behavioral correctness theorem for the compiler.
What \Cref{thm:compiler-wp-pres} adds is the form of its conclusion, a
weakest precondition of the target logic, which composes with target-level
reasoning. In \Cref{sec:secomp}, it is what lets a compiled function
satisfy a target contract that other compiled modules rely on
(\Cref{thm:sep-fun-contract}).

\begin{example}[Fibonacci]\label{ex:fib}
Consider the following implementation of Fibonacci
\begin{center}
\small
$\begin{array}{@{}l@{}}
 \mathsf{FIB}(n) :\\
 \quad c:=1; \; x:=0; \; y:=1;\\
 \quad\mathsf{while}\ c\le n\ \mathsf{do}\
   (\mathit{temp}:=y; \; y:=\mathit{temp}+x; \; x:=\mathit{temp}; \;c:=c+1);\\
 \quad\mathit{fib}:=x
\end{array}$
\end{center}
and, writing $\mathrm{Fib}(k)$ for the $k$-th Fibonacci number, its specification
\[
  {\src\Phi}:=
     c\mathrel{\src{\mapsto}}n{+}1 *
     x\mathrel{\src{\mapsto}}\mathrm{Fib}(n) *
     y\mathrel{\src{\mapsto}}\mathrm{Fib}(n{+}1) *
     \mathit{fib}\mathrel{\src{\mapsto}}\mathrm{Fib}(n).
\]
We proved $P_{\mathsf{FIB}}\vdash\wpIMP{\mathsf{FIB}(n)}{{\src\Phi}}$,
where $P_{\mathsf{FIB}}$ owns $c$, $x$, $y$, $\mathit{temp}$ and
$\mathit{fib}$ with arbitrary, possibly uninitialized, values.
Instantiating~\Cref{thm:compiler-wp-pres} with $\mathsf{FIB}(n)$ and
extracting the points-to assertions of ${\src\Phi}$
(\Cref{thm:compiler-wp-pres-frame}) gives the target specification
\[ \Corr_{\mathsf{FIB}(n)}*P_{\mathsf{FIB}}\;\vdash\;
  \wpi{\trgH}{\denoteASM{\mathsf{compile}(\mathsf{FIB}(n))}(0)}{{\trg\Psi}}
\]
where
\[
\begin{aligned}
  {\trg\Psi}:={}&
     a_c\mathrel{\trg{\mapsto}}n{+}1 *
     a_x\mathrel{\trg{\mapsto}}\mathrm{Fib}(n) *
     a_y\mathrel{\trg{\mapsto}}\mathrm{Fib}(n{+}1) *
     a_{\mathit{fib}}\mathrel{\trg{\mapsto}}\mathrm{Fib}(n)\\
  &{}*\mathsf{RegOwn}(X_{\mathsf{FIB}(n)})*\mathsf{Cell}(\mathit{temp}).
\end{aligned}
\]
The precondition owns only the cells of the five variables of
$\mathsf{FIB}(n)$ and the registers of its compiled code. The postcondition
consists of target points-to assertions, the registers, and the
correspondence cell of $\mathit{temp}$, whose final value ${\src\Phi}$ does
not describe.

\smallskip
Since the precondition owns no other memory cell, the frame rule of the
target logic applies. For a variable $z$ that $\mathsf{FIB}(n)$ does not use,
whose target cell holds $42$,
$
  \Corr_{\mathsf{FIB}(n)}*P_{\mathsf{FIB}}*a_z\mathrel{\trg{\mapsto}}42
  \;\vdash\;
  \wpi{\trgH}{\denoteASM{\mathsf{compile}(\mathsf{FIB}(n))}(0)}{%
      {\trg\Psi}*a_z\mathrel{\trg{\mapsto}}42}.
$
\end{example}

\section{Separate Compilation}\label{sec:secomp}
We extend IMP, ASM and the compiler with modules and function calls, and
show that WP-preservation survives separate compilation and linking. Each
compiled function is related to its source by an open simulation, which
yields the preservation of function and module contracts and, once the
modules are linked, of the WP of the whole program
(\Cref{sec:secomp-wp-pres}). The transfer lemma
of~\Cref{sec:abstract-transfer} is reused unchanged. We illustrate the
results with two modules whose mutually recursive functions pass ownership
of a shared counter (\Cref{sec:secomp-counter}). \Cref{app:modules} in the 
appendix gives the full details.

\subsection{Modules, Calls, and Linking}\label{sec:secomp-syntax}
\begin{figure}[htbp]
\centering
\small
$
\begin{array}{@{}l@{\;}c@{\;}l@{\quad}l@{\;}c@{\;}l@{}}
M & ::= & \mathsf{module}\ N\ \{\vec{g};\,\vec{F};\,\mathsf{exp}\ \vec{\sigma};\,\mathsf{imp}\ \vec{\sigma}\}
  & \sigma & ::= & f/n\\[.4ex]
F_{\mathrm{IMP}} & ::= & f\mapsto\langle \vec{p},\vec{l},c,a\rangle
  & F_{\mathrm{ASM}} & ::= & f\mapsto\langle n,\vec{l},\mathit{regs},B\rangle\\[.4ex]
c & ::= & \dots\mid x:=\mathsf{call}\ f(\vec{a})
  & i & ::= & \dots\mid \mathsf{acall}\ r,f,\vec{o}\\[.4ex]
T & ::= & \{N.f\mapsto F,\dots\}=\mathsf{link}(\vec{M})
  & \mathit{program} & ::= & c\ \text{run with}\ T
\end{array}
$
\caption{Modules and calls, simplified (full definitions in
Appendix~\ref{app:modules}). A module $N$ declares globals $\vec{g}$,
functions $\vec{F}$, and exported ($\mathsf{exp}$) and imported
($\mathsf{imp}$) signatures $f/n$ (name and arity). An IMP function $f$ has
formal parameters $\vec{p}$, locals $\vec{l}$, a body $c$ and a result
expression $a$, evaluated after the body; ASM functions have an arity $n$, locals $\vec{l}$, the
registers $\mathit{regs}$ they use, and a body $B:\mathsf{asm}(1,1)$. Variables are
local, $y$, or global, $N.y$; we omit the module qualification of called
functions. A function table $T$, produced by linking modules $\vec{M}$, maps
the qualified name $N.f$ of each function to its definition; a program is an
entry command run with a table.}
\label{fig:modules-syntax}
\end{figure}
\paragraph{Modules.}
IMP and ASM programs are organized in modules of the same shape, parametric
in the kind of function they contain (\Cref{fig:modules-syntax}). A module has a name, a list of global
variables, a list of named functions, and the signatures (name and arity)
of the functions it exports and imports. Variables are either local, $x$,
or global, $M.x$, where $M$ is the module that declares $x$; memory
addresses follow the same scheme, and the address $a_x$ of a variable $x$
is the address of the same kind and name.

\paragraph{Functions and calls.}
An IMP function has a list of parameters, a list of local variables, a body
command, and a result expression evaluated after the body; there is no
return command. We extend IMP with the command
$x:=\mathsf{call}\ f(a_1,\dots,a_n)$, which first evaluates the arguments,
calls the function $f$, and assigns its result to $x$.
An ASM function has an arity, a list of local variables, the set of
registers it uses, and a body, an $\mathsf{asm}(1,1)$ graph whose
single exit returns to the caller, with the result in register $r_0$.
ASM gains a call instruction $\mathsf{acall}\ r,f,o_1,\dots,o_n$, which
passes the operands $o_i$ to $f$ and writes the result to $r$.

\paragraph{Function call events and linking.}
In both languages a call triggers the event $\mathsf{FnCall}(f,\vec{v})$,
which carries the callee and the argument values and returns the result.
For IMP, the call command then writes the result to $x$;
for ASM, the call instruction writes it to $r$. Calls are
thus visible events of the \itree of a function body, 
later interpreted by linking.

Linking a list of modules produces a \emph{function table}, a finite map
from qualified names $N.f$, for each function $f$ of a module $N$, to
function definitions. We write $\srcT$ for the table of the source modules
and $\trgT$ for the table of their compilation. A
\emph{resolver} then interprets the $\mathsf{FnCall}$ events of a program
recursively (using \texttt{interp\_mrec} from the \itree library). Each call performs a step, looks up the
callee in the table, checks the number of arguments, runs the body of the
callee, and returns its result. A call whose callee is missing, or has the
wrong arity, results in an undefined behavior in IMP and $\mathsf{Exit}$ in ASM.
A \emph{program} is thus an entry command $c$ run with a function table
$T$, that is, the resolver of $T$ runs $c$ and resolves each call by running the
body that $T$ assigns to the callee. The program is \emph{closed} when
every function called by $c$, or by a function of $T$, is defined in $T$.
The closed programs of~\Cref{sec:wp-pres} are the special case of a
command without calls and an empty table.

\paragraph{Calling convention.}
The caller $x:=\mathsf{call}\ f(a_1,\dots,a_n)$ is compiled to code that
evaluates each argument $a_i$ into register $r_{i-1}$, followed by
\[
  \mathsf{acall}\ r_0,f,r_0,\dots,r_{n-1};\ \mathsf{store}\ a_x,r_0,
\]
which calls $f$, receives the result in $r_0$, and stores it at the address
of $x$. The compiled callee starts with a prologue that stores
its parameters from $r_0,\dots,r_{n-1}$ to their local addresses, and ends
by computing its result expression into $r_0$.
The remaining parts of the convention are carried out by the target
resolver, which saves the values of the callee's local addresses, resets
them, writes the argument values to $r_0,\dots,r_{n-1}$, runs the body,
restores the saved values, and reads the result from $r_0$. ASM has no
activation records. A local address $a_x$ is a single memory cell, shared by
all functions that use the name $x$, and the resolver saves and restores
these cells so that nested calls do not overwrite the locals of their
callers. The source resolver performs the
same steps on the store, event for event, so that the two resolvers can be
related step by step.

\paragraph{Separate compilation.}
A module is compiled function by function. Hence $\trgT$ can be obtained either by linking the compiled modules or by
compiling each function of $\srcT$, and both yield the same table.

\subsection{Contracts}\label{sec:secomp-contracts}

A \emph{contract} $\gamma=\langle\mathsf{pre},\mathsf{post}\rangle$
specifies a single function by a precondition $\mathsf{pre}(a,\vec{v})$ over
its argument values and a postcondition $\mathsf{post}(a,r)$ over its result.
Both mention a logical witness $a$, chosen by the caller, which lets the
postcondition refer to what the caller handed over, e.g.\ to return a
resource with the value it had. A \emph{contract context} $\Gamma$ is a
partial map from function names to contracts, and $\Gamma(f)$ is the
contract for $f$, which callers rely on and which $f$ must be proved to
satisfy. 

Contexts give call events their logical meaning. 
For any $\mathsf{FnCall}(f,\vec{v})$, the call handler
$\handler{call}(\Gamma)$ looks up $\Gamma(f)=\langle\mathsf{pre},\mathsf{post}\rangle$,
gives away $\mathsf{pre}(a,\vec{v})$ for some witness $a$, and receives
$\mathsf{post}(a,r)$ for the result $r$; only calls to functions with a contract can be verified (\Cref{app:sep-call-rules}).

\paragraph{Contracts for separately compiled modules.}
For a list of linked modules, we distinguish between calls to functions
defined by one of them (\emph{internal} calls) and calls to functions defined
elsewhere (\emph{external} calls).
Contracts for internal calls are collected in $\srcGint$,
contracts for external calls in $\Gamma_{\mathit{ext}}$.
In the source, internal calls also lend the caller's local variables:
the ownership $\mathsf{LocalsAt}(W)$ of the locals $X$ of the link at
values $W$ known to the caller, which the callee returns with the same
values. Locals are shared cells that a callee may overwrite, and the
resolver saves and restores them around each call, so a caller keeps the
values of its locals across calls; we call such contracts \emph{scoped}.
External contracts are the same in source and target, while the target
contracts of internal calls additionally carry the state correspondence
$\Corr$ from caller to callee and back, similar in spirit
to~\Cref{thm:compiler-wp-pres} (see~\Cref{app:sep-contracts} for more
detail).\footnote{Throughout this section $\Corr$ is
$\Corr_{\mathit{Pool},V}=\mathsf{RegOwn}(\mathit{Pool})*\Corr_V$
(\Cref{sec:imp-asm-simulation}), which generalizes $\Corr_c$
of~\Cref{sec:wp-pres}: it owns a pool $\mathit{Pool}$ of registers shared
by all linked functions instead of the registers of a single command, and
the cells of a set $V$ of \emph{visible} variables, fixed for the link,
instead of the variables of a single command.}

\paragraph{Well-formed modules.}
A module is well formed if each function with a contract in
$\srcGint$ has a source body whose compilation is exactly its
entry in $\trgT$, and that body is proved against its contract in
the source logic, with scoped contracts for internal calls; the remaining
conditions concern arities and variable names (\Cref{app:sep-wf}). 

\subsection{WP-Preservation under Separate Compilation}\label{sec:secomp-wp-pres}

We write $\mathsf{run}_{T}(f,\vec{v})$ for one call of $f$ by the resolver
of a table $T$, which saves and resets the locals of $f$, passes $\vec{v}$,
runs the body of $f$, restores the locals, and returns the result, and
$\mathsf{link}_{T}(t)$ for the \itree $t$ whose calls to functions of $T$
are resolved recursively by $T$.

We write $\mathsf{link}^{\partial}_{T}(t)$ for the \itree obtained from $t$
by resolving its calls to functions of $T$ recursively, and leaving every
other call as an external $\mathsf{FnCall}$ event (\emph{partial}
resolution).

We write $\trgH_{\Gamma}$ for the ASM
logical handler whose calls are answered by $\handler{call}(\Gamma)$. 

We extend the notation $\wpIMP{c}{{\src\Phi}}$ of~\Cref{sec:poc} in two ways.
For an open program, whose calls are answered by the contracts of a source
context $\srcGamma$,
\[
 \wpIMPx{\srcGamma}{c}{{\src\Phi}}\;:=\;
 \wpi{\handler{call}(\srcGamma)\oplus\srcH}{\denoteIMP{c}}{{\src\Phi}},
\]
and for the program $c$ run with the source table $\srcT$,
\[
 \wpIMPx{\srcT}{c}{{\src\Phi}}\;:=\;
 \wpi{\srcH}{\mathsf{link}_{\srcT}(\denoteIMP{c})}{{\src\Phi}}.
\]
For a command without calls, both coincide with $\wpIMP{c}{{\src\Phi}}$.

We organize the results below as Compositional
CompCert~\cite{compositionalCompcert} organizes its correctness proof, with a
simulation for each compiled function, from which the per-module results and
a linking theorem follow. Unlike there, the results transport weakest
preconditions and contracts rather than behaviors.

\paragraph{Functions.}
\begin{lemma}[Open simulation and resolution]\label{lem:sep-fun-sim}
\leavevmode
\begin{enumerate}
\item \emph{Per function.} Each source function, run by the resolver of
$\srcT$, is related to its compiled function, run by the resolver of $\trgT$:
from $\Corr$, the two runs return equal results together with $\Corr$,
and the calls they make remain related events.
\item \emph{Under resolution.} If the source and target resolvers are
related at every call, resolving the calls of two related \itrees
recursively yields related \itrees.
\end{enumerate}
\end{lemma}
The precise statements, with their side conditions on registers and
variables, are~\Cref{lem:sep-fun-sim-formal,lem:sep-linking-law}
in~\Cref{app:modules}.

For a source contract $\gamma=\langle\mathsf{pre},\mathsf{post}\rangle$, the
target contract
$\compProp{\gamma}=\langle\compProp{\mathsf{pre}},\compProp{\mathsf{post}}\rangle$
adds the scoped resources and the state correspondence
(\Cref{sec:secomp-contracts}). Its witness is a pair $(W,a)$ of the
caller's locals and the witness of $\gamma$, and
\[
\begin{aligned}
 \compProp{\mathsf{pre}}((W,a),\vec{v})&:=
   \Corr*\ulcorner\mathrm{dom}\,W=X\urcorner*\mathsf{LocalsAt}(W)*\mathsf{pre}(a,\vec{v}),\\
 \compProp{\mathsf{post}}((W,a),r)&:=
   \mathsf{post}(a,r)*\mathsf{LocalsAt}(W)*\Corr,
\end{aligned}
\]
with
$\mathsf{LocalsAt}(W):=\mathop{\scalebox{1.5}{$\ast$}}_{x\mapsto w\in W}x\mathrel{\src{\mapsto^{?}}}w$,
and $\srcGamma$ answers internal calls to $f$ by the scoped contract of
$\srcGint(f)$ and external calls by $\Gamma_{\mathit{ext}}$,
while $\trgGamma$ answers them by $\compProp{\srcGint(f)}$
and $\Gamma_{\mathit{ext}}$.

\begin{theorem}[Function contract preservation]\label{thm:sep-fun-contract}
Let $f$ have body $\mathit{fn}$ in $\srcT$ and the compiled body
$\mathsf{compile}(\mathit{fn})$ in $\trgT$, with registers in $\mathit{Pool}$
and visible variables, and with contract
$\srcGint(f)=\langle\mathsf{pre},\mathsf{post}\rangle$, and
let its body be proved against this contract in the source logic, with
scoped contracts for internal calls. For all argument values $\vec{v}$ of
the right length and all witnesses $(W,a)$,
\[
 \compProp{\mathsf{pre}}((W,a),\vec{v})
 \;\vdash\;
 \wpi{\trgH_{\trgGamma}}{\mathsf{run}_{\trgT}(f,\vec{v})}{\compProp{\mathsf{post}}((W,a),\cdot)}.
\]
\end{theorem}
In words, the compiled code of $f$ satisfies its target contract, assuming
that all the functions it calls satisfy their own target contracts. The proof
applies~\Cref{lem:abstract-wp-transfer}, unchanged, to the simulation
of~\Cref{lem:sep-fun-sim} and to the source WP of $\mathsf{run}_{\srcT}(f,\vec{v})$
obtained from the contract proof of the body (\Cref{app:sep-proofs}).

\paragraph{Modules.}
Applying~\Cref{thm:sep-fun-contract} to every function of a module gives
the result at the granularity of modules.

\begin{corollary}[Per-module contract preservation]\label{cor:sep-module}
Let $M$ be a well-formed module. Then~\Cref{thm:sep-fun-contract} holds for
every function $f$ of $M$ with a contract in $\srcGint$.
\end{corollary}
The corollary concerns $M$ alone, since each call of $M$ to a function of another
module is answered by the contract of that function, from
$\srcGint$ if the module is linked with $M$ and from
$\Gamma_{\mathit{ext}}$ otherwise.

\paragraph{Linked programs.}
Finally, we state the result for linked modules. Let $M_1,\dots,M_k$ be
the source modules, linked into $\srcT$ and compiled into $\trgT$, and let
$c$ be an entry command such that $c$ and every function of $\srcT$
have their registers in $\mathit{Pool}$ and their variables visible. A
source postcondition ${\src\Phi}$ for $c$ can be established against the
table $\srcT$ or against the contracts $\srcGamma$:

\begin{theorem}[Separate Compilation WP-preservation]\label{thm:sep-linking}
Let $M_1,\dots,M_k$ be well formed. For every source 
${\src\Phi}$, and for
$\compProp{{\src\Phi}}:=\Corr*{\src\Phi}(\mathsf{tt}):$
\begin{enumerate}
\item \emph{Separate compilation preserves the weakest preconditions of the
source program closed by linking,}
\[
 \Corr*\wpIMPx{\srcT}{c}{{\src\Phi}}
 \;\vdash\;
 \wpi{\trgH}{\mathsf{link}_{\trgT}(\denoteASM{\mathsf{compile}(c)}(0))}{\compProp{{\src\Phi}}}.
\]
\item \emph{Separate compilation preserves the weakest preconditions proved
against the contracts in $\srcGamma$,}
\[
 \Corr*\wpIMPx{\srcGamma}{c}{{\src\Phi}}
 \;\vdash\;
 \wpi{\trgH_{\Gamma_{\mathit{ext}}}}{\mathsf{link}^{\partial}_{\trgT}(\denoteASM{\mathsf{compile}(c)}(0))}{\compProp{{\src\Phi}}}.
\]
\end{enumerate}
\end{theorem}

In both items the modules are compiled separately. In item~1 the source
proof concerns the linked, closed program; item~2 makes verification
separate as well. Each module is
verified once against the contracts (\Cref{cor:sep-module}) and the entry
command against the same contracts, so no proof looks at the linked source
program.
Item~1 follows from~\Cref{lem:sep-fun-sim}. By its first part the
resolvers are related at every call, so by its second part the two linked
programs are related, and~\Cref{lem:abstract-wp-transfer}
applies once. Item~2 follows from~\Cref{cor:sep-module} and L\"ob induction
over the steps performed by the resolver at each call. Both proofs are
in~\Cref{app:sep-closed}. 

\begin{remark}\label{rem:sep-target}
As in~\Cref{sec:wp-pres}, in both items the points-to assertions of
${\src\Phi}$ on visible variables can be extracted to target points-to
assertions. The shared counter of~\Cref{sec:secomp-counter} shows this, as its
final target postcondition contains $a_{G.c}\mathrel{\trg{\mapsto}}0$.
Moreover, an adequacy corollary discharges the source and target logical
handlers against the computational interpretation of their events, for the
compiled linked program run from initial states related by $\Corr$
(\Cref{cor:sep-adequacy}).
\end{remark}

\subsection{Case Study: a Shared Counter}\label{sec:secomp-counter}

We link three modules: $G$ declares a global counter $c$, and $M_1$ and
$M_2$ define functions $f$ and $g$ that call each other until $c$ reaches
zero,
\begin{center}
\small
$\begin{array}{@{}l@{}}
 \mathsf{module}\ G\ \{c;\,\epsilon;\,\mathsf{exp}\ \epsilon;\,\mathsf{imp}\ \epsilon\}\\[.4ex]
 \mathsf{module}\ M_1\ \{\epsilon;\,f\mapsto\langle\epsilon,[k,r],c_f,0\rangle;\,
   \mathsf{exp}\ M_1.f/0;\,\mathsf{imp}\ M_2.g/0\}\\[.4ex]
 \mathsf{module}\ M_2\ \{\epsilon;\,g\mapsto\langle\epsilon,[k,r],c_g,0\rangle;\,
   \mathsf{exp}\ M_2.g/0;\,\mathsf{imp}\ M_1.f/0\}
\end{array}$
\end{center}
where $\epsilon$ is the empty list and
\[
 c_f:=\mathsf{if}\ G.c=0\ \mathsf{then}\ \mathsf{skip}\
   \mathsf{else}\ (\mathsf{PRED}(G.c,k);\ r:=\mathsf{call}\ M_2.g()),
\]
and $c_g$ is the same with $M_1.f$ as callee. Here $\mathsf{PRED}(x,k)$
decrements $x$ by counting $k$ up to $x-1$ in a loop, and the result of
each call is assigned to the unused local $r$.
The entry command is $\mathit{main}:=r:=\mathsf{call}\ M_1.f()$.

Both functions have the same contract, which passes ownership of the
counter to the callee and returns it set to zero:
\[
 \srcGint(M_1.f)=\srcGint(M_2.g)
 =\bigl\langle\lambda\vec{v}.\,\ulcorner\vec{v}=[\,]\urcorner*
   \exists n.\,G.c\mathrel{\src{\mapsto}}n,\;
   \lambda\_.\,G.c\mathrel{\src{\mapsto}}0\bigr\rangle,
\]
and $\Gamma_{\mathit{ext}}$ is empty, since $\mathit{main}$ run with the
table of the three modules is a closed program.
Each body is proved against this contract in the source logic, assuming the contract
of the other function for the call it makes (\textsc{Call-Scoped}
in~\Cref{app:sep-call-rules}). The proofs of $f$ and $g$ are independent,
and the two modules are well formed. By~\Cref{cor:sep-module}, the compiled
code of each of them satisfies its target contract.

By item~2 of~\Cref{thm:sep-linking}, once the three compiled modules are
linked, the compiled entry command ends with the counter set to zero in
target memory:
\[
\begin{aligned}
 &\Corr*\mathsf{LocalPool}(\{k,r\})*
   \exists n.\,G.c\mathrel{\src{\mapsto}}n\\
 &\quad\vdash\;
 \wpi{\trgH_{\emptyset}}{\mathsf{link}^{\partial}_{\trgT}(\denoteASM{\mathsf{compile}(\mathit{main})}(0))}%
     {a_{G.c}\mathrel{\trg{\mapsto}}0*{\trg\theta}},
\end{aligned}
\]
where ${\trg\theta}$ collects the registers of $\mathit{Pool}$, the
correspondence cells of $k$ and $r$, and the source pool
$\mathsf{LocalPool}(\{k,r\})$ of $\mathit{main}$, whose locals are owned
with arbitrary values. The precondition owns only the cells of $k$, $r$
and $G.c$, not the rest of target memory. Both
functions use the locals $k$ and $r$, so their calls overwrite the same
memory cells $a_k$ and $a_r$; the save and restore performed by the
resolver keep each activation's values separate.

\section{Related Work}\label{sec:relatedW}

\paragraph{Specification-preserving compilation.}
Sch\"afer et~al.~\cite{axsemcomp} define the semantics of a
guarded-command and a goto language by total-correctness judgments
$\langle\sigma\rangle\,s\,\langle Q\rangle$ and prove in Rocq that their
compiler preserves these judgments; quantified over all $Q$, preservation
coincides with contextual approximation.
Omnisemantics~\cite{omnisemantics} generalize this judgment to the
omni-big-step judgment and state correctness as $t\Downarrow Q$ implies
$C(t)\Downarrow Q$; quantified over all $Q$, this is termination-preserving
behavioral refinement, as the authors prove by instantiating $Q$ with the
strongest postcondition. Across languages $Q$ is lifted through a state
relation $R$ to $\exists\mathit{src}.\,R\wedge Q$, the shape of our
$\compProp{\Phi}=\Corr*\Phi$; the Bedrock2 compiler is verified this
way~\cite{integrationVerification}.
In both lines $Q$ is a pure predicate over concrete states, $R$ relates
whole source and target states, and the results hold for terminating
programs only. Our $\Corr$ and $\Phi$ are Iris assertions over different
states and logical handlers. $\Corr_c$ owns only the cells of the variables
of $c$ and the registers of its compiled code, so the transported
specification keeps a small footprint and the target frame rule applies
(\Cref{ex:fib}); it carries ownership and ghost state, accounts for
registers that have no source counterpart, composes with other target-level
Iris proofs, covers divergence, and extends to contracts under separate
compilation.

\paragraph{Proof-transforming compilation and certificate translation}
take source derivations as input and transport the proof of one
specification. Barthe et~al. show that a non-optimizing compiler preserves
generated proof obligations exactly~\cite{obligationsPres}, and transport
certificates through an imperative-to-RTL compiler with optimizations by
strengthening annotations with certified analysis
results~\cite{certTransComp}; Barthe and Kunz~\cite{certTransAbs}
generalize the approach to an abstract notion of certificate translation. Saabas and
Uustalu~\cite{compositionalNaturalSemantics} compile Hoare proofs of While
programs, together with the programs, into proofs for a low-level language
with jumps. Nordio et~al. compile Eiffel contracts and Hoare-style proofs to
.NET CIL~\cite{proofTrans} and translate separation-logic derivations for
object-oriented programs into bytecode proofs~\cite{certTransNordio}.
Seo et~al.~\cite{potpie} compile programs, specifications and Hoare proofs
together, with a Rocq proof that the proof compiler is sound whenever the
compiler satisfies a few equational conditions, even if it is incorrect
elsewhere. All of these translate syntactic objects, proofs or
verification conditions, into a target logic designed to mirror the source
one; only Seo et~al. report a mechanization. Our compiler operates on
programs alone and preserves every semantic WP, however established; the
transfer relates different event signatures and logical handlers, accounts
for ownership, and extends to separate compilation. Like the compiler
of~\cite{obligationsPres}, ours does not optimize, but we require a
simulation rather than identical proof obligations.

\paragraph{Behavioral preservation and separate compilation.}
Semantics preservation is the criterion of CompCert~\cite{compcert}
and CakeML~\cite{cakeML}: every behavior of the compiled program is a
behavior of the source. Compositional CompCert~\cite{compositionalCompcert}
extends it to separately compiled modules through structured simulations
with module-local invariants, and CompCertO~\cite{CompCertO} through
simulation conventions that relate source and target interfaces, calling
conventions included. Song et~al.~\cite{conditionalContextualRefinement}
verify \itree modules against abstract ones under separation-logic
conditions on the modules they call, and reach assembly through CompCert;
Zhang et~al.~\cite{threeDimensionalRefinement} organize such results in a
refinement algebra along modules, abstraction levels and state components.
In all of these, what is preserved is behavioral refinement, and
conditions, where present, live at the source, so the compiled code receives
no contract, whereas \Cref{thm:sep-fun-contract} derives one for each
compiled function from its source contract.

\paragraph{End-to-end verification.}
The Verified Software Toolchain connects Verifiable C, a separation logic
for CompCert's Clight, to CompCert through adequacy, as logical proofs imply
properties of the source semantics and CompCert's correctness theorem
transfers those behavioral guarantees to assembly~\cite{VST}. Mansky and
Du rebuild Verifiable C on Iris~\cite{vstiris}, and Archmage follows the
same architecture for programs that cast between integers and pointers,
with a memory model shared by the logic and the CompCertCast
compiler~\cite{archmage}. In all three, the guarantee at the assembly
level is behavioral, and there is no specification in a logic of the target
that a target-level proof, or another compiled module, could use.
WP-preservation yields such a specification, and in \Cref{sec:secomp} a
target contract for each compiled function.

\paragraph{Interaction trees and compiler verification.}
Xia et~al.~\cite{itrees} prove their IMP-to-ASM compiler correct as weak
bisimilarity between source and target \itree denotations interpreted into
a common state monad; Silver et~al.~\cite{itreeNI} extend the result to
noninterference with an event-sensitive bisimulation. Vellvm scales the
approach to a substantial subset of LLVM IR~\cite{vellvmExperience}, and
Torrini and Gr\'egoire~\cite{itreescompil} to more than twenty passes of
the Jasmin front end, with relational Hoare-logic proofs of equivalence
\emph{up to cutoff}, which relates source undefined behavior with arbitrary
target behavior; probabilistic \itree semantics later extend this to the
whole Jasmin compiler and to preservation of KEM-IND-CCA
security~\cite{jasminITrees}. DimSum~\cite{dimsum} specifies each language
independently and relates them through protocols, with a separation logic
for modular refinement proofs. All these results are bisimulations or
refinements of behaviors; none transports a WP to the target logic.

\paragraph{Proof techniques similar to ours.}
Simuliris~\cite{simuliris} internalizes a simulation over small-step
semantics in a non-step-indexed Iris logic, with ownership reasoning for
fair termination-preserving contextual refinement of concurrent
optimizations within one language. Its simulation nests a least fixed point
for stuttering inside a greatest fixed point, as our $\nu G.\,\mu L$ does
for administrative target steps; FreeSim~\cite{freesim} generalizes such
stuttering to progress recorded on each side independently. The Tail Modulo
Cons proof~\cite{tailmodulocons} builds on Simuliris to verify a compiler
pass for all inputs. It relates calls of source and target through
protocols, as our call row does, and hands ownership of an uninitialized
destination field to the target callee, recovering it filled with a value
related to the source result; we do the same with $\Corr$ across matched
events and calls, against logical handlers rather than an operational state
interpretation. These works establish behavioral refinement; our simulation
is defined over \itrees, at the semantic level of the WP calculus, so that
each step is tied to the logical specification of an event, and it
establishes WP-preservation across a lowering.

\paragraph{Secure compilation.}
Van Strydonck et~al.~\cite{lcap4fac} compile separation-logic-verified
code using linear capabilities and boundary checks, and prove their compiler
fully abstract~\cite{fac,surveySECOMP}. A compiler is fully abstract when
it preserves and reflects contextual equivalence. Full abstraction,
however, is known not to guarantee that a property of an individual source
program, and hence a specification, holds after
compilation~\cite{beyondFA,cond2fac,logicalFAC}.
Robust compilation~\cite{beyondFA,tracerel,phdAbate} instead preserves
properties of traces against adversarial contexts. The framework of
Abate et al.~\cite{tracerel} in principle allows traces to carry information
about resource ownership, but robust compilation has so far been achieved
for traces of I/O events~\cite{goodCompBad,SECOMP,misquoted}.
M\"uck et~al.~\cite{robustdimsum} use DimSum to verify a hardening compiler
to a capability machine, proving that Reckon specifications are robustly
safe after compilation.
We do not consider adversarial contexts, but preserve separation-logic
specifications, including ownership, as weakest preconditions of the
compiled code.

\section{Conclusions and Future Work}\label{sec:conclusions}

We investigated WP-preserving compilation as a way of transporting
source-level logical interfaces to compiled code. Unlike trace-based
semantics preservation~\cite{compcert}, WP-preservation does not relate 
the observable behaviors of source and target programs, and unlike compilers that
preserve predicates over whole final states~\cite{axsemcomp,omnisemantics},
it transports judgments of a program logic, with ownership and a footprint
limited to what the compiled code uses. It
relates specifications provable in a separation logic for the source
program, including ownership of resources, to specifications that hold for
the compiled target program.
Our main proof technique is a resource-aware simulation over interaction
trees (\Cref{sec:wp-pres}). The generic WP-transfer lemma isolates the
properties that a simulation must satisfy to transport weakest preconditions
between \itrees with potentially different event signatures, return types,
and logical handlers.
We instantiated this simulation for the IMP and ASM languages and proved that
the denotations of every IMP program and its compiled ASM program are related.
This case study demonstrates that logical specifications can survive compilation
from a structured source language to a control-flow-graph language, even when
the two semantics expose different effects and the target performs additional
administrative operations.
The target precondition owns only the memory cells of the variables of the
compiled command and the registers of its code, so target-level clients can
frame the rest of target memory.
Finally, we extended WP-preservation from closed programs to separately
compiled modules (\Cref{sec:secomp}): each module is verified in the source
logic against the contracts of its imports and compiled separately, and its
specifications carry over to the compiled module and are preserved when
compiled modules are linked, and callers keep the values of their locals
across calls.

\paragraph{Limitations.}
For separate compilation, \Cref{thm:sep-linking} covers closed programs
whose target modules are all produced by our compiler, while open modules
are covered only through contracts (\Cref{cor:sep-module}), relative to a
register pool, local pool and set of visible variables fixed for the whole
link. In particular, the target contract of a call says that the callee may
overwrite every register of the pool, so target-level clients cannot keep
values in registers across calls. ASM has no activation records: the
resolver saves and restores the locals of a callee, and this saving is not
proved correct with respect to a stack-allocated semantics.

\paragraph{Future work.} 
A first, mild generalization concerns the state correspondence, which
relates each source variable to one target cell holding the same value. A
value relation $\rho$ and a location map $\ell$ would do, with
$\mathsf{Cell}_\rho(x):=\exists w,u.\,\mathsf{Auth}_x(w)*\ell(x)\mathrel{\trg{\mapsto^{?}}}u*\ulcorner\rho(w,u)\urcorner$:
the target postcondition $\Corr_\rho*{\src\Phi}$ is then the image of
${\src\Phi}$ along the induced relation between stores and memories, in the
sense of trace-relating compilation~\cite{tracerel}, and a source
$x\mathrel{\src{\mapsto}}v$ yields
$\exists u.\,\ulcorner\rho(v,u)\urcorner*\ell(x)\mathrel{\trg{\mapsto}}u$ in
the target. The transfer lemma is unaffected; the event relations of the
simulation and the points-to extraction would change, and a location map
that depends on the program point would be needed for register allocation.
Similarly, points-to assertions are extracted to the target only at full
ownership; extracting a fraction $x\mathrel{\src{\mapsto_{q}}}v$ as
$a_x\mathrel{\trg{\mapsto_{q}}}v$ would let $\mathsf{Cell}(x)$ hand out
shares of its target cell, a change confined to the extraction lemmas.

More ambitious future work concerns the generic transfer lemma, whose
conditions we satisfy with a simulation specific to the events of IMP and
ASM. A parameterized resource-aware \itree relation, in the style of
\texttt{eutt}, \texttt{rutt}~\cite{itrees} or \texttt{xrutt}~\cite{itreescompil},
would provide the transfer, bind and iteration principles once,
and let the results for individual passes compose into an end-to-end theorem.

WP-preservation for separate compilation currently links compiled modules
only. Extending it to arbitrary, possibly adversarial target contexts, i.e.\ to
the robust preservation of safety specifications~\cite{beyondFA}, requires
restricting what target modules may do, for instance by a static check at
link time or by a capability mechanism in ASM~\cite{robustdimsum}.

\subsubsection*{Acknowledgments}
This work has been supported by funding from the Agentur f\"ur Innovation in
der Cybersicherheit GmbH (Cyberagentur).

\clearpage 
\bibliographystyle{splncs04}
\bibliography{paper}

\begin{thebibliography}{10}
\providecommand{\url}[1]{\texttt{#1}}
\providecommand{\urlprefix}{URL }
\providecommand{\doi}[1]{https://doi.org/#1}

\bibitem{fac}
Abadi, M.: Protection in programming-language translations. In: International
  Colloquium on Automata, Languages, and Programming. pp. 868--883. Springer
  (1998)

\bibitem{phdAbate}
Abate, C.: A formal framework for correct and secure compilation. {PhD} thesis,
  Ruhr-Universit{\"a}t Bochum, Universit{\"a}tsbibliothek (2023).
  \doi{10.13154/294-9870}

\bibitem{goodCompBad}
Abate, C., Azevedo~de Amorim, A., Blanco, R., Evans, A.N., Fachini, G.,
  Hri{\c{t}}cu, C., Laurent, T., Pierce, B.C., Stronati, M., Tolmach, A.: When
  good components go bad: Formally secure compilation despite dynamic
  compromise. In: Proceedings of the 2018 ACM SIGSAC Conference on Computer and
  Communications Security. pp. 1351--1368 (2018)

\bibitem{tracerel}
Abate, C., Blanco, R., Ciob{\^a}c{\u{a}}, {\c{S}}., Durier, A., Garg, D.,
  Hri{\c{t}}cu, C., Patrignani, M., Tanter, {\'E}., Thibault, J.: An extended
  account of trace-relating compiler correctness and secure compilation. ACM
  Transactions on Programming Languages and Systems (TOPLAS)  \textbf{43}(4),
  1--48 (2021)

\bibitem{beyondFA}
Abate, C., Blanco, R., Garg, D., Hri{\c{t}}cu, C., Patrignani, M., Thibault,
  J.: Journey beyond full abstraction: Exploring robust property preservation
  for secure compilation. In: 2019 IEEE 32nd Computer Security Foundations
  Symposium (CSF). pp. 256--271. IEEE (2019)

\bibitem{tailmodulocons}
Allain, C., Bour, F., Cl{\'e}ment, B., Pottier, F., Scherer, G.: Tail modulo
  cons, {OCaml}, and relational separation logic. Proc. ACM Program. Lang.
  \textbf{9}(POPL) (Jan 2025). \doi{10.1145/3704915}

\bibitem{misquoted}
Andrici, C.C., Pribisova, A., Ahman, D., Hri{\c{t}}cu, C., Rivas, E.,
  Winterhalter, T.: Misquoted no more: Securely extracting {F*} programs with
  {IO}. Proceedings of the ACM on Programming Languages  \textbf{10}(ICFP),
  494--523 (2026)

\bibitem{VST}
Appel, A.W.: Program Logics for Certified Compilers. Cambridge University Press
  (2014)

\bibitem{jasminITrees}
Arranz-Olmos, S., Barthe, G., Blatter, L., Gr{\'e}goire, B., Laporte, V.,
  Torrini, P.: {KEM-IND-CCA}-preserving compilation of {Jasmin}'s {ML-KEM}
  (2026), \url{https://arxiv.org/abs/2511.11292}

\bibitem{certTransComp}
Barthe, G., Gr{\'e}goire, B., Kunz, C., Rezk, T.:
  Certificate translation for optimizing compilers. In: Yi, K. (ed.) Static
  Analysis. pp. 301--317. Springer Berlin Heidelberg, Berlin, Heidelberg (2006)

\bibitem{certTransAbs}
Barthe, G., Kunz, C.: An abstract model of certificate translation. ACM Trans.
  Program. Lang. Syst.  \textbf{33}(4) (Jul 2011).
  \doi{10.1145/1985342.1985344}

\bibitem{obligationsPres}
Barthe, G., Rezk, T., Saabas, A.: Proof obligations preserving compilation. In:
  International Workshop on Formal Aspects in Security and Trust. pp. 112--126.
  Springer (2005)

\bibitem{vellvmExperience}
Beck, C., Chen, H., Zdancewic, S.: {Vellvm}: Formalizing the informal {LLVM}
  (experience report). In: NASA Formal Methods ({NFM} 2025). Lecture Notes in
  Computer Science, vol. 15682, pp. 91--99. Springer (2025).
  \doi{10.1007/978-3-031-93706-4_6}

\bibitem{omnisemantics}
Chargu{\'e}raud, A., Chlipala, A., Erbsen, A., Gruetter, S.: Omnisemantics:
  Smooth handling of nondeterminism. {ACM} Transactions on Programming
  Languages and Systems  \textbf{45}(1),  5:1--5:43 (2023).
  \doi{10.1145/3579834}

\bibitem{freesim}
Cho, M., Song, Y., Lee, D., G{\"a}her, L., Dreyer, D.: Stuttering for free.
  Proc. ACM Program. Lang.  \textbf{7}(OOPSLA2) (Oct 2023).
  \doi{10.1145/3622857}

\bibitem{dijkstra75}
Dijkstra, E.W.: Guarded commands, nondeterminacy and formal derivation of
  programs. Communications of the ACM  \textbf{18}(8),  453--457 (1975)

\bibitem{integrationVerification}
Erbsen, A., Gruetter, S., Choi, J., Wood, C., Chlipala, A.: Integration
  verification across software and hardware for a simple embedded system. In:
  Proceedings of the 42nd {ACM} {SIGPLAN} International Conference on
  Programming Language Design and Implementation ({PLDI} 2021). pp. 604--619.
  {ACM} (2021). \doi{10.1145/3453483.3454065}

\bibitem{simuliris}
G{\"a}her, L., Sammler, M., Spies, S., Jung, R., Dang, H.H., Krebbers, R.,
  Kang, J., Dreyer, D.: Simuliris: a separation logic framework for verifying
  concurrent program optimizations. Proc. ACM Program. Lang.  \textbf{6}(POPL)
  (Jan 2022). \doi{10.1145/3498689}

\bibitem{iris}
Jung, R., Krebbers, R., Jourdan, J., Bizjak, A., Birkedal, L., Dreyer, D.:
  {Iris} from the ground up: {A} modular foundation for higher-order concurrent
  separation logic. J. Funct. Program.  \textbf{28}, ~e20 (2018).
  \doi{10.1017/S0956796818000151}

\bibitem{archmage}
Kim, Y., Cho, M., Lee, J., Kim, J., Yoon, T., Song, Y., Hur, C.K.: Archmage and
  {CompCertCast}: End-to-end verification supporting integer-pointer casting.
  Proc. ACM Program. Lang.  \textbf{9}(POPL) (Jan 2025). \doi{10.1145/3704881}

\bibitem{CompCertO}
Koenig, J., Shao, Z.: {CompCertO}: compiling certified open {C} components. In:
  Proceedings of the 42nd ACM SIGPLAN International Conference on Programming
  Language Design and Implementation. pp. 1095--1109. PLDI 2021, Association
  for Computing Machinery, New York, NY, USA (2021).
  \doi{10.1145/3453483.3454097}

\bibitem{compcert}
Leroy, X.: Formal certification of a compiler back-end or: programming a
  compiler with a proof assistant. In: Morrisett, J.G., Jones, S.L.P. (eds.)
  Proceedings of the 33rd {ACM} {SIGPLAN-SIGACT} Symposium on Principles of
  Programming Languages, {POPL} 2006, Charleston, South Carolina, USA, January
  11-13, 2006. pp. 42--54. {ACM} (2006). \doi{10.1145/1111037.1111042}

\bibitem{logicalFAC}
Longley, J.R., Plotkin, G.: Logical full abstraction and {PCF}. In: Proof,
  Language, and Interaction: Essays in Honour of Robin Milner. MIT Press (2000)

\bibitem{vstiris}
Mansky, W., Du, K.: An {Iris} instance for verifying {CompCert} {C} programs.
  Proc. {ACM} Program. Lang.  \textbf{8}({POPL}),  148--174 (2024).
  \doi{10.1145/3632848}

\bibitem{robustdimsum}
M{\"u}ck, N., Georges, A.L., Dreyer, D., Garg, D., Sammler, M.: Endangered by
  the language but saved by the compiler: Robust safety via semantic
  back-translation. Proc. ACM Program. Lang.  \textbf{10}(POPL) (Jan 2026).
  \doi{10.1145/3776682}

\bibitem{certTransNordio}
Nordio, M., Calcagno, C., Meyer, B.: Certificates and separation logic. In:
  International Symposium on Trustworthy Global Computing. pp. 273--293.
  Springer (2013)

\bibitem{proofTrans}
Nordio, M., M{\"u}ller, P., Meyer, B.: Proof-transforming compilation of
  {Eiffel} programs. In: International Conference on Objects, Components,
  Models and Patterns. pp. 316--335. Springer (2008)

\bibitem{seplogicOH}
O'Hearn, P.: Separation logic. Commun. ACM  \textbf{62}(2),  86--95 (Jan 2019).
  \doi{10.1145/3211968}

\bibitem{seplogicOH07}
O’Hearn, P.W.: Resources, concurrency, and local reasoning. Theoretical
  Computer Science  \textbf{375}(1),  271--307 (2007).
  \doi{10.1016/j.tcs.2006.12.035}, festschrift for John C. Reynolds’s 70th
  birthday

\bibitem{seplogicParkinson}
Parkinson, M., Bierman, G.: Separation logic and abstraction. In: Proceedings
  of the 32nd ACM SIGPLAN-SIGACT symposium on Principles of programming
  languages. pp. 247--258 (2005)

\bibitem{cond2fac}
Parrow, J.: General conditions for full abstraction. Mathematical Structures in
  Computer Science  \textbf{26}(4),  655--657 (2016)

\bibitem{colors}
Patrignani, M.: Why should anyone use colours? or, syntax highlighting beyond
  code snippets (2022), \url{https://squera.github.io/misc/colours-1.pdf},
  accessed: 2026-09-01

\bibitem{surveySECOMP}
Patrignani, M., Ahmed, A., Clarke, D.: Formal approaches to secure compilation:
  A survey of fully abstract compilation and related work. ACM Computing
  Surveys (CSUR)  \textbf{51}(6),  1--36 (2019)

\bibitem{seplogicReynolds}
Reynolds, J.: Separation logic: a logic for shared mutable data structures. In:
  Proceedings 17th Annual IEEE Symposium on Logic in Computer Science. pp.
  55--74 (2002). \doi{10.1109/LICS.2002.1029817}

\bibitem{compositionalNaturalSemantics}
Saabas, A., Uustalu, T.: A compositional natural semantics and {Hoare} logic
  for low-level languages. Theoretical Computer Science  \textbf{373}(3),
  273--302 (2007). \doi{10.1016/j.tcs.2006.12.020}

\bibitem{dimsum}
Sammler, M., Spies, S., Song, Y., D'Osualdo, E., Krebbers, R., Garg, D.,
  Dreyer, D.: {DimSum}: {A} decentralized approach to multi-language semantics
  and verification. Proc. {ACM} Program. Lang.  \textbf{7}({POPL}),  775--805
  (2023). \doi{10.1145/3571220}

\bibitem{axsemcomp}
Sch{\"a}fer, S., Schneider, S., Smolka, G.: Axiomatic semantics for compiler
  verification. In: Proceedings of the 5th {ACM} {SIGPLAN} Conference on
  Certified Programs and Proofs ({CPP} 2016). pp. 188--196. {ACM} (2016).
  \doi{10.1145/2854065.2854083}

\bibitem{potpie}
Seo, A., Lam, C., Grossman, D., Ringer, T.: Correctly compiling proofs about
  programs without proving compilers correct. In: 15th International Conference
  on Interactive Theorem Proving ({ITP} 2024). LIPIcs, vol.~309, pp.
  33:1--33:20. Schloss Dagstuhl -- Leibniz-Zentrum f{\"u}r Informatik (2024).
  \doi{10.4230/LIPIcs.ITP.2024.33}

\bibitem{itreeNI}
Silver, L., He, P., Cecchetti, E., Hirsch, A.K., Zdancewic, S.: Semantics for
  noninterference with interaction trees. In: 37th European Conference on
  Object-Oriented Programming (ECOOP 2023). Leibniz International Proceedings
  in Informatics (LIPIcs), vol.~263, pp. 29:1--29:29. Schloss
  Dagstuhl--Leibniz-Zentrum fuer Informatik (2023).
  \doi{10.4230/LIPIcs.ECOOP.2023.29}

\bibitem{conditionalContextualRefinement}
Song, Y., Cho, M., Lee, D., Hur, C., Sammler, M., Dreyer, D.: Conditional
  contextual refinement. Proceedings of the {ACM} on Programming Languages
  \textbf{7}({POPL}),  1121--1151 (2023). \doi{10.1145/3571232}

\bibitem{compositionalCompcert}
Stewart, G., Beringer, L., Cuellar, S., Appel, A.W.: Compositional {CompCert}.
  In: Proceedings of the 42nd Annual ACM SIGPLAN-SIGACT Symposium on Principles
  of Programming Languages. pp. 275--287 (2015)

\bibitem{cakeML}
Tan, Y.K., Myreen, M.O., Kumar, R., Fox, A.C.J., Owens, S., Norrish, M.: The
  verified {CakeML} compiler backend. J. Funct. Program.  \textbf{29}, ~e2
  (2019). \doi{10.1017/S0956796818000229}

\bibitem{SECOMP}
Thibault, J., Blanco, R., Lee, D., Argo, S., Azevedo~de Amorim, A., Georges,
  A.L., Hri{\c t}cu, C., Tolmach, A.: {SECOMP}: Formally secure compilation of
  compartmentalized {C} programs. In: Proceedings of the 2024 on ACM SIGSAC
  Conference on Computer and Communications Security. pp. 1061--1075. CCS '24,
  Association for Computing Machinery, New York, NY, USA (2024).
  \doi{10.1145/3658644.3670288}

\bibitem{itreescompil}
Torrini, P., Gr{\'e}goire, B.: Interaction trees and verified compilation. In:
  The Rocq Workshop (2025),
  \url{https://coq-workshop.gitlab.io/2025/files/EA5.pdf}

\bibitem{lcap4fac}
Van~Strydonck, T., Piessens, F., Devriese, D.: Linear capabilities for fully
  abstract compilation of separation-logic-verified code. Journal of Functional
  Programming  \textbf{31}, ~e6 (2021). \doi{10.1017/S0956796821000022}

\bibitem{plalacarte}
Vistrup, M., Sammler, M., Jung, R.: Program logics {\`a} la carte. Proceedings
  of the ACM on Programming Languages  \textbf{9}(POPL),  300--331 (2025)

\bibitem{itrees}
Xia, L.y., Zakowski, Y., He, P., Hur, C.K., Malecha, G., Pierce, B.C.,
  Zdancewic, S.: Interaction trees: representing recursive and impure programs
  in {Coq}. Proceedings of the ACM on Programming Languages  \textbf{4}(POPL),
  1--32 (2020)

\bibitem{threeDimensionalRefinement}
Zhang, Y., Koenig, J., Shao, Z., Wang, Y.: Unifying compositional verification
  and certified compilation with a three-dimensional refinement algebra.
  Proceedings of the {ACM} on Programming Languages  \textbf{9}({POPL}),
  1903--1933 (2025). \doi{10.1145/3704900}

\end{thebibliography}

\clearpage 
\appendix
\crefname{appendix}{Appendix}{Appendices}
\Crefname{appendix}{Appendix}{Appendices}
\crefalias{section}{appendix}
\crefalias{subsection}{appendix}
\renewcommand{\theHsection}{appendix.\Alph{section}}
\makeatletter
\@for\thm@ctr:=theorem,lemma,definition,corollary,proposition\do{%
  \counterwithin*{\thm@ctr}{section}%
  \setcounter{\thm@ctr}{0}%
  \expandafter\edef\csname the\thm@ctr\endcsname{\noexpand\thesection\noexpand\arabic{\thm@ctr}}%
  \expandafter\edef\csname theH\thm@ctr\endcsname{appendix.\noexpand\thesection.\noexpand\arabic{\thm@ctr}}%
}
\makeatother
\section{Details of WP-Preservation}\label{app:wp-pres}

This appendix gives the transfer interface, simulation definition, and
resource assertions used in~\Cref{sec:wp-pres}.

\subsection{The abstract transfer interface}\label{app:abstract-transfer}

Let $R:{\src A}\to{\trg A}\to\iProp\;\Sigma$ relate return values,
and let $\mathcal{E}(e,f)$ and $\mathcal{A}(e,a,f,b)$ be Iris-valued
relations on events and their answers. All assertions use the same Iris
resource universe. Write $W_H^M(t,\Phi)$ for the itree WP with handler $H$
and invariant mask $M$; $W_H$ abbreviates $W_H^\emptyset$.
For the visible-event case, fix a source answer type $X_s$, an event
$e:{\src E}\,X_s$, and a continuation
$k:X_s\to\itreeER{\src E}{\src A}$.
We write $H(e,Q;\theta)$ for the logical handler applied to the event,
its continuation assertion $Q$, and its auxiliary thread-spawning
continuation assertion $\theta$. The latter is part of the library's
general handler interface; the IMP and ASM handlers ignore it.
For every target postcondition $\Psi$, the abstract simulation must satisfy:
\[
\tag{\textsc{SIMret}}\label{eq:sim-ret}
 \SIM_R(\Ret{u},t)
 *\bigl(\forall v.\;R(u,v)\wand W_{\trg H}(\Ret{v},\Psi)\bigr)
 \vdash W_{\trg H}(t,\Psi),
\]
\[
\tag{\textsc{SIMtau}}\label{eq:sim-tau}
 \SIM_R(\Tau{s},t)\vdash\SIM_R(s,t),
\]
\[
\tag{\textsc{SIMvis}}\label{eq:sim-vis}
\begin{aligned}
 &\SIM_R(\Vis{e}{k},t)\\
 &*\Bigl(\forall (X_t:\mathsf{Type})\,
       (f:{\trg E}\,X_t)\,
       (h:X_t\to\itreeER{\trg E}{\trg A}).\\
 &\qquad\mathcal{E}(e,f)\wand\\
 &\qquad\bigl(\forall (a:X_s)\,(b:X_t).\;\mathcal{A}(e,a,f,b)
                    \wand\SIM_R(k(a),h(b))\bigr)\wand\\
 &\qquad\Bigl(\forall\Psi',\theta_s,\theta_t.\\
 &\qquad\quad {\src H}\bigl(e,\lambda a.\,\forall b.\,
       \mathcal{A}(e,a,f,b)\wand\Psi'(b);\theta_s\bigr)\\
 &\qquad\quad\wand\mathsf{Upd}_{\emptyset}\,{\trg H}(f,\Psi';\theta_t)\Bigr)\wand\\
 &\qquad W_{\trg H}(\Vis{f}{h},\Psi)\Bigr)\\
 &\qquad\vdash W_{\trg H}(t,\Psi).
\end{aligned}
\]
Here $X_t$ is the target event's answer type,
$f:{\trg E}\,X_t$, and $h:X_t\to\itreeER{\trg E}{\trg A}$.
The assertions $\Psi',\theta_t:X_t\to\iProp\;\Sigma$ and
$\theta_s:X_s\to\iProp\;\Sigma$ range over continuation assertions;
$a:X_s$ and $b:X_t$ are event answers, and $\mathsf{Upd}_{\emptyset}$ is the
Iris fancy update at the empty mask. It lets the transport perform ghost
updates: with handlers specified by points-to assertions only, no handler
clause offers an update modality, and the matched write of IMP and ASM
updates the source authority of the written variable (\Cref{sec:imp-asm-simulation}).
The requirement is weaker than one without the update, so the lemma is
stronger.

The premise quantified over $X_t,f,h$ asks the caller to prove the WP of
the matched target event, given three resources: the event relation,
the continuation simulation, and the displayed handler-transport wand.
When establishing \textsc{SIMvis} for IMP and ASM, the simulation proof
constructs this wand from the resources of the matched-event clause
(in particular, $\mathcal{T}(e,f)$ in~\Cref{fig:imp-asm-simulation}).
It also accounts for the target administrative steps preceding the match.

In the abstract transfer proof, the wand is received as
\texttt{Htransport}. Applying the target visible-event WP rule reduces
the goal to a target handler assertion. The proof then applies
\texttt{Htransport}: handler monotonicity, the source WP induction
hypothesis, and the continuation simulation establish its source-handler
premise, yielding the required target-handler assertion.
Thus the transport wand connects the two handlers at the matched event,
as described in~\Cref{sec:abstract-transfer}.

\begin{lemma}[Abstract transfer with masks]\label{lem:abstract-wp-transfer-masks}
Suppose the simulation satisfies the three requirements above.
For invariant masks $M_s\subseteq M_t$,
\[
\begin{aligned}
 &\SIM_R(s,t)*W_{\src H}^{M_s}(s,\Phi)\\
 &*\bigl(\forall u,v.\;R(u,v)*\Phi(u)
       \wand\mathsf{Upd}_{M_s,M_t}(\Psi(v))\bigr)\\
 &\qquad\vdash W_{\trg H}^{M_t}(t,\Psi).
\end{aligned}
\]
Here $\mathsf{Upd}_{M_s,M_t}(P)$ denotes an Iris fancy update from mask
$M_s$ to $M_t$. Taking equal masks and using update
introduction gives~\Cref{lem:abstract-wp-transfer}.
\end{lemma}

\subsection{The IMP--ASM simulation}\label{app:imp-asm-simulation}

For $R:{\src A}\to{\trg A}\to\iProp\;\Sigma$, the simulation is
\[
 \simIMPASM_R:=\nu G.\,\mu L.\,\simFunctor(G,L),
\]
with the functor and auxiliary assertions in~\Cref{fig:imp-asm-simulation}.
The simulation resources are
$\Corr_{X,V}:=\mathsf{RegOwn}(X)*\Corr_V$ (\Cref{sec:imp-asm-simulation}).
For a command $c$, let $X_c$ be the compiler's structural register
footprint and $V_c$ the variables of $c$, write $\Corr_c:=\Corr_{X_c,V_c}$,
and take $R_c(u,v):=\Corr_c$. The compiler simulation of~\Cref{sec:imp-asm-wp-pres}
is proved for $\Corr_{X,V}$ with any $X\supseteq X_c$ and $V\supseteq V_c$.

\subsection{Composition and iteration}\label{app:sim-algebra}

\begin{lemma}[Composition properties]\label{lem:sim-algebra-full}
The IMP--ASM simulation satisfies the abstract transfer interface
in~\Cref{app:abstract-transfer}. It also has the following properties.
\begin{enumerate}
 \item If $s'\approx s$ and $t'\approx t$, then
 $\simIMPASM_R(s,t)\vdash\simIMPASM_R(s',t')$.
 \item For an intermediate return relation $Q$,
 \[
 \begin{aligned}
 &\simIMPASM_Q(s,t)
 *\bigl(\forall a,b.\;Q(a,b)\wand\simIMPASM_R(k(a),h(b))\bigr)\\
 &\qquad\vdash\simIMPASM_R(\mathsf{bind}(s,k),\mathsf{bind}(t,h)).
 \end{aligned}
 \]
 \item Let $I$ relate loop states, and define
 \[
 (I+R)(x,y):=
 \begin{cases}
 I(i,j)&x=\mathsf{inl}(i),\ y=\mathsf{inl}(j),\\
 R(u,v)&x=\mathsf{inr}(u),\ y=\mathsf{inr}(v),\\
 \mathsf{False}&\text{otherwise}.
 \end{cases}
 \]
 If the relation
 \[
 i'\prec i\quad:\!\Longleftrightarrow\quad
 f(i)\approx\Ret{\mathsf{inl}(i')}
 \]
 is well founded, then
 \[
 \begin{aligned}
 &\square\bigl(\forall i',j'.\;I(i',j')\wand
        \simIMPASM_{I+R}(f(i'),g(j'))\bigr)*I(i,j)\\
 &\qquad\vdash\simIMPASM_R(\mathsf{iter}(f,i),\mathsf{iter}(g,j)).
 \end{aligned}
 \]
 Here $\mathsf{inl}$ restarts the loop and $\mathsf{inr}$ returns;
 $\square$ makes the body premise reusable. The well-foundedness
 condition is on source restarts without a visible event.
\end{enumerate}
\end{lemma}

\subsection{Specification transfer and target points-to}\label{app:compiler-wp-pres}

\Cref{tab:ownership} summarizes who holds each resource.

\begin{table}[htbp]
\centering
\footnotesize
\begin{tabular}{@{}llll@{}}
\toprule
resource & source WP & simulation & target WP \\
\midrule
$x\mathrel{\src{\mapsto}}v$ & client & --- & extracted as $a_x\mathrel{\trg{\mapsto}}v$ \\
$\mathsf{Auth}_x(w)$, $x\in V_c$ & --- & $\Corr_c$ & $\Corr_c$ \\
$a_x\mathrel{\trg{\mapsto^{?}}}w$, $x\in V_c$ & --- & $\Corr_c$ & $\Corr_c$ \\
$\mathsf{RegOwn}(X_c)$ & --- & $\Corr_c$ & $\Corr_c$ \\
other target cells & --- & --- & client (target frame) \\
\bottomrule
\end{tabular}
\caption{Who holds each resource in the source WP, in the state
correspondence of the simulation, and in the target WP; a dash means the
resource does not occur. The authorities over the whole store, memory and
register file occur only in the adequacy proofs of the handlers.}
\label{tab:ownership}
\end{table}

\begin{theorem}[WP preservation: full statement]
\label{thm:compiler-wp-pres-full}
For every command $c$, source postcondition ${\src\Phi}$, invariant
mask $M$, and the same step modality on both sides,
\[
\begin{aligned}
 &\mathsf{RegOwn}(X_c)*\Corr_{V_c}*W^{M}_{\src H}(\denoteIMP{c},{\src\Phi})\\
 &\quad\vdash W_{\trg H}^{M}
   \bigl(\denoteASM{\mathsf{compile}(c)}(0),\lambda\_.\,
      \mathsf{RegOwn}(X_c)*\Corr_{V_c}*{\src\Phi}(\mathsf{tt})\bigr),
\end{aligned}
\]
for every target exit postcondition.
\end{theorem}
For a finite map $G$ of initialized variables, abbreviate
\[
 {\src P}(G):=\mathop{\scalebox{1.5}{$\ast$}}_{x\mapsto v\in G}
             x\mathrel{\src{\mapsto}}v,
 \qquad
 {\trg P}(G):=\mathop{\scalebox{1.5}{$\ast$}}_{x\mapsto v\in G}
             a_x\mathrel{\trg{\mapsto}}v.
\]

\begin{lemma}[Extracting target points-to assertions]\label{lem:pointsto-transfer}
\[
 x\mathrel{\src{\mapsto}}v*\mathsf{Cell}(x)\vdash a_x\mathrel{\trg{\mapsto}}v,
 \qquad
 {\src P}(G)*\Corr_{\mathrm{dom}\,G}\vdash{\trg P}(G).
\]
\end{lemma}
\begin{proof}
The authority of $x$ in $\mathsf{Cell}(x)$ agrees with
$x\mathrel{\src{\mapsto}}v$, so the entry $w$ shared with the target cell is
$\mathsf{Some}(v)$. The second entailment is the first, cell by cell.\qed
\end{proof}

\begin{theorem}[WP preservation with target points-to]
\label{thm:compiler-wp-pres-frame}
For every command $c$, initialized map $F$ with
$\mathrm{dom}\,F\subseteq V_c$, and assertion $Q$,
\[
\begin{aligned}
 &\Corr_c*\wpIMP{c}{\lambda\_.\,{\src P}(F)*Q}\\
 &\quad\vdash\wpi{\trgH}{\denoteASM{\mathsf{compile}(c)}(0)}%
   {\mathsf{RegOwn}(X_c)*{\trg P}(F)*Q*\Corr_{V_c\setminus\mathrm{dom}\,F}}.
\end{aligned}
\]
Source and target use the same invariant mask and step modality, and the
target exit postcondition is arbitrary.
\end{theorem}
\begin{proof}[Proof sketch]
\Cref{thm:compiler-wp-pres} returns ${\src P}(F)*Q$ with $\Corr_c$. Split
$\Corr_{V_c}$ into $\Corr_{\mathrm{dom}\,F}*\Corr_{V_c\setminus\mathrm{dom}\,F}$
and apply~\Cref{lem:pointsto-transfer}. For derivable source Hoare triples,
source soundness supplies the source WP premise.\qed
\end{proof}
A target assertion about a cell outside $V_c$ is preserved by the frame rule
of the target logic, since the precondition does not own that cell.

\section{The IMP Program Logic}\label{app:imp-logic}

This appendix presents the WP rules for IMP in detail.
The rules are collected in~\Cref{fig:imp-logic-rules}. We fix an invariant
mask $M$ and a step modality $m$, and use the following abbreviations to
keep the figure compact:
\[
\begin{aligned}
 W_a(a,\Phi)&:=\wpi{\srcH}{\mathsf{denote\_aexp}(a)}{\Phi},\\
 W_b(b,\Phi)&:=\wpi{\srcH}{\mathsf{denote\_bexp}(b)}{\Phi},\\
 W(c,\Phi)&:=\wpi{\srcH}{\denoteIMP{c}}{\Phi}.
\end{aligned}
\]
All three itree WPs use mask $M$ and
$\srcH=\handler{store}\oplus\handler{step}(m)\oplus\handler{ub}$.
Arithmetic and Boolean expressions return natural numbers and Booleans,
respectively; commands return $\mathsf{tt}$.

We write $\triangleright_m P$ for the modality selected by $m$:
$\triangleright_{\mathsf{Later}}P=\triangleright P$ and
$\triangleright_{\mathsf{Identity}}P=P$.
The assertion $x\mathrel{\src{\mapsto^{?}}}w$, where
$w:\mathsf{option}(\mathbb N)$, permits an uninitialized value;
$x\mathrel{\src{\mapsto}}v$ abbreviates the case $w=\mathsf{Some}(v)$.
Both denote full ownership; the variable rule also holds for a fraction,
$x\mathrel{\src{\mapsto_{q}}}v$, which the figure omits.
All ownership arrows in the figure refer to source resources.
The rules are Iris entailments. A rule with premises $P_1,\dots,P_k$ and
conclusion $Q$ stands for $P_1*\dots*P_k\vdash Q$: its premises are
resources supplied together, not separate proof obligations. A wand in a
premise describes how to use resources returned by a subproof.

\begin{figure}[p]
\centering
\footnotesize
\textbf{Arithmetic and Boolean expressions.}
\begin{mathpar}
 \inferrule*[right=Const]
   {\Phi(n)}
   {W_a(n,\Phi)}
 \and
 \inferrule*[right=Var]
   {x\mathrel{\src{\mapsto}}v \\ x\mathrel{\src{\mapsto}}v\wand\Phi(v)}
   {W_a(x,\Phi)}
 \and
 \inferrule*[right=Arith]
   {W_a(a_1,\lambda v_1.\,W_a(a_2,\lambda v_2.\,\Phi(v_1\circ v_2)))}
   {W_a(a_1\circ a_2,\Phi)}
 \and
 \inferrule*[right=Bool]
   {\Phi(\beta)}
   {W_b(\beta,\Phi)}
 \and
 \inferrule*[right=Compare]
   {W_a(a_1,\lambda v_1.\,W_a(a_2,\lambda v_2.\,\Phi(v_1\bowtie v_2)))}
   {W_b(a_1\bowtie a_2,\Phi)}
 \and
 \inferrule*[right=Not]
   {W_b(b,\lambda\beta.\,\Phi(\neg\beta))}
   {W_b(\neg b,\Phi)}
 \and
 \inferrule*[right=And]
   {W_b(b_1,\lambda\beta_1.\,W_b(b_2,\lambda\beta_2.\,
       \Phi(\beta_1\land\beta_2)))}
   {W_b(b_1\land b_2,\Phi)}
\end{mathpar}
Here $\circ\in\{+,\times\}$, $\bowtie\in\{=,\leq\}$ is Boolean-valued,
and $\beta\in\{\mathsf{true},\mathsf{false}\}$.

\medskip
\textbf{Commands.}
\begin{mathpar}
 \inferrule*[right=Skip]
   {P}
   {W(\mathsf{skip},\lambda\_.\,P)}
 \and
 \inferrule*[right=Seq]
   {W(c_1,\lambda\_.\,W(c_2,\Phi))}
   {W(c_1;c_2,\Phi)}
 \and
 \inferrule*[right=Assign]
   {x\mathrel{\src{\mapsto^{?}}}w \\
    x\mathrel{\src{\mapsto^{?}}}w\wand
       W_a(a,\lambda v.\,\Theta(v)*x\mathrel{\src{\mapsto^{?}}}w) \\
    \forall v.\;\Theta(v)*x\mathrel{\src{\mapsto}}v\wand Q}
   {W(x:=a,\lambda\_.\,Q)}
 \and
 \inferrule*[right=If]
   {W_b(b,\lambda\beta.\,\mathsf{if}\ \beta\ \mathsf{then}\ W(c_1,\Phi)
                    \ \mathsf{else}\ W(c_2,\Phi))}
   {W(\mathsf{if}\ b\ \mathsf{then}\ c_1\ \mathsf{else}\ c_2,\Phi)}
 \and
 \inferrule*[right=While]
   {W_b\bigl(b,\lambda\beta.\,\triangleright_m
       (\mathsf{if}\ \beta\ \mathsf{then}\
        W(c,\lambda\_.\,W(\mathsf{while}\ b\ \mathsf{do}\ c,\Phi))
        \ \mathsf{else}\ \Phi(\mathsf{tt}))\bigr)}
   {W(\mathsf{while}\ b\ \mathsf{do}\ c,\Phi)}
\end{mathpar}

\medskip
\textbf{Structural and modality rules.}
\begin{mathpar}
 \inferrule*[right=Consequence]
   {\forall v.\;\Phi(v)\wand\Psi(v) \\ W_\star(e,\Phi)}
   {W_\star(e,\Psi)}
 \and
 \inferrule*[right=Frame]
   {P \\ W(c,\Phi)}
   {W(c,\lambda u.\,P*\Phi(u))}
 \and
 \inferrule*[right=Mask]
   {\mathsf{Upd}_{M_1,M_2}\bigl(
       W^{m,M_2}(c,\lambda u.\,\mathsf{Upd}_{M_2,M_1}(\Phi(u)))\bigr)}
   {W^{m,M_1}(c,\Phi)}
 \and
 \inferrule*[right=Weaken]
   {W^{\mathsf{Identity},M}(c,\Phi)}
   {W^{\mathsf{Later},M}(c,\Phi)}
\end{mathpar}
Here $W_\star$ ranges over $W_a,W_b,W$ with the appropriate expression
or command type. Superscripts make the normally implicit modality and mask
explicit; $\mathsf{Upd}_{M_1,M_2}$ is the Iris fancy update between those masks.
\caption{Weakest precondition calculus for IMP.}
\label{fig:imp-logic-rules}
\end{figure}

\paragraph{Reading the rules.}
Expression evaluation proceeds from left to right. In particular,
\textsc{And} evaluates both operands. The variable rule preserves ownership
of the variable being read. The assignment rule lets the right-hand side
use the assigned variable's ownership, provided that ownership is returned
with its old value before the write. The auxiliary assertion $\Theta(v)$
carries the other resources produced by expression evaluation into the
postcondition proof. This supports assignments such as $x:=x+1$.
The sequencing rule passes the first command's resources directly to the
second command.
The while rule exposes one guard evaluation and retains the step
modality around the chosen continuation. With $m=\mathsf{Later}$,
this permits guarded recursive reasoning.

\paragraph{Ghost state of stores, registers and memory.}
The store, the register file and the memory are finite maps from keys $K$
(variables, registers, addresses) to optional values, accessed one key at a
time. They share one ghost state, defined once for a key type $K$ and a
value type $V$ and built only from Iris library components: a resource
$\gamma$ in the camera
$\mathsf{gmap}(K,\mathsf{auth}(\mathsf{option}(\mathsf{dfrac\_agree}(\mathsf{option}(V)))))$,
with one authority per key and fractions on the fragments, and a ghost map
$\gamma_{\mathsf{dom}}$ of the keys present:
\[
\begin{aligned}
 k\mathrel{\mapsto^{?}_{q}}w&:=
   \mathsf{own}\,\gamma\,\{k\mapsto\circ\,\mathsf{Some}(q,w)\}
   *k\in_{\gamma_{\mathsf{dom}}}\!\square,\\
 \mathsf{Auth}_k(w)&:=
   \mathsf{own}\,\gamma\,\{k\mapsto\bullet\,\mathsf{Some}(1,w)\}.
\end{aligned}
\]
Instantiated at variables, these are the source assertions
$x\mathrel{\src{\mapsto^{?}_{q}}}w$ and $\mathsf{Auth}_x(w)$; at registers
and addresses, the target assertions. We prove that
$\mathsf{Auth}_k(w)*k\mathrel{\mapsto^{?}_{q}}w'$ implies $w=w'$ for every
fraction $q$, that fragments split and agree as usual, and that
$\mathsf{Auth}_k(w)*k\mathrel{\mapsto^{?}}w$ can be updated to
$\mathsf{Auth}_k(w')*k\mathrel{\mapsto^{?}}w'$. Because the authority is kept
per key, the authorities of distinct keys compose; this is what lets
$\Corr_V$ of~\Cref{sec:imp-asm-simulation} own the authorities of the
variables in $V$ and nothing else, whereas a single authority, as in a
standard ghost map, would make $\Corr_V$ range over every variable.

The handlers, the program logics and the adequacy proofs see this ghost
state only through points-to assertions and the interpretation
\[
 S(\sigma):=\mathsf{DomAuth}(\mathrm{dom}\,\sigma)
   *\mathop{\scalebox{1.5}{$\ast$}}_{k\mapsto w\in\sigma}\mathsf{Auth}_k(w),
\]
where $\mathsf{DomAuth}$ is the authority of $\gamma_{\mathsf{dom}}$; the
simulation also uses the per-key authorities, with agreement and joint
update; the camera itself appears only at allocation, which yields
$S(\sigma)$ and a full points-to assertion for every key of $\sigma$. One
adequacy proof serves the three handlers: the per-key events are
implemented by events that read and write the whole map, and a WP over the
former yields a WP over the latter, with $S(\sigma)$ split in two halves
between the state interpretation and the proof. An absent key is a
parameter of the implementation: undefined behavior for the store, which
the domain witness of an owned variable rules out, and a read of
$\mathsf{None}$ and an insertion on write for registers and memory. The
target uses the per-key authorities too, although our simulation needs only
target points-to, so that a later pass with ASM as its source can relate
registers and memory cells as $\mathsf{Cell}$ relates variables.

\clearpage 
\section{Details of Separate WP-Preservation}\label{app:modules}

This appendix gives the technical details of~\Cref{sec:secomp}.

\subsection{Modules, Linking, and Compilation}\label{app:sep-modules}

\lstdefinelanguage{rocqmini}{
  morekeywords={Record, Variant, Type},
  sensitive=true,
  morecomment=[s]{(*}{*)}
}

The following definitions, from the Rocq development, give the full syntax
of modules and functions summarized in~\Cref{fig:modules-syntax}. A module
is parametric in the type \texttt{F} of its functions: \texttt{F} is
\texttt{function} for IMP and \texttt{asm\_function} for compiled ASM.

\begin{lstlisting}[language=rocqmini, basicstyle=\footnotesize\ttfamily, columns=fullflexible]
Variant var  := VLocal (x : string)
              | VGlobal (mn : module_name) (x : string).
Variant addr := Local (x : string)
              | Global (mn : module_name) (x : string).
Record qualname  := { qual_mod : ident; qual_fun : ident }.
Record signature := { sig_name : qualname; sig_arity : nat }.

Record function := {
  fun_params : list ident;  fun_locals : list ident;
  fun_body   : com;         fun_result : aexp }.
Record asm_function := {
  af_arity : nat;           af_locals : list string;
  af_regs  : gset reg;      af_body   : asm 1 1 }.

Record module (F : Type) := {
  mod_name      : ident;
  mod_globals   : list ident;          (* names only *)
  mod_functions : list (ident * F);
  mod_exports   : list signature;
  mod_imports   : list signature }.
\end{lstlisting}

Linking a list of modules qualifies each function name $f$ of a module $N$
as $N.f$ and takes the union of the resulting function tables.
Compiling a module maps the compiler over its functions: an IMP function
is compiled to the ASM function whose arity is the number of parameters,
whose locals are the parameters followed by the locals, whose registers are
those used by the compiled code, and whose body is the compiled body.

\paragraph{Resolvers.}
A table $T$ is resolved by a recursive interpretation of call events. One
call $\mathsf{FnCall}(f,\vec{v})$ performs a step and looks up $f$ in $T$. If
$f$ is missing, or its arity differs from the length of $\vec{v}$, the call
is undefined behavior in IMP and $\mathsf{Exit}$ in ASM. Otherwise the
resolver runs $\mathsf{run}_T(f,\vec{v})$, which
\begin{enumerate}
 \item saves the current values of the parameters and locals of $f$,
 \item resets them to unset values,
 \item passes the arguments, in IMP by assigning them to the parameters and
       in ASM by writing them to registers $r_0,\dots,r_{n-1}$,
 \item runs the body of $f$, which may perform further calls,
 \item computes the result, in IMP by evaluating the result expression and
       in ASM by the compiled body, which leaves it in $r_0$,
 \item restores the saved values, and returns the result, in ASM read
       from $r_0$.
\end{enumerate}
In ASM, the saved values are kept in the \itree, not in memory, so the
local cells of $f$ are free during the call. The source resolver performs
the same steps on the store, in the same order, so that the two resolvers
emit corresponding events.
The \emph{partial} resolver of a target table $\trgT$ handles calls to
functions of $\trgT$ as above, and re-emits every other call unchanged, as an
external event answered by the environment.

\paragraph{Compiled functions.}
The body of a compiled function is
\[
 \mathsf{store}\ a_{p_1},r_0;\ \dots;\ \mathsf{store}\ a_{p_n},r_{n-1};\
 \mathsf{compile}(c);\ \mathsf{compile}_{r_0}(a),
\]
for parameters $p_1,\dots,p_n$, body $c$ and result expression $a$, where
$\mathsf{compile}_{r_0}(a)$ evaluates $a$ into $r_0$. A call command
$x:=\mathsf{call}\ f(a_1,\dots,a_n)$ is compiled to the evaluation of each
$a_i$ into $r_{i-1}$, followed by
$\mathsf{acall}\ r_0,f,r_0,\dots,r_{n-1};\ \mathsf{store}\ a_x,r_0$.

\begin{lemma}[Linking commutes with compilation]\label{lem:link-compile}
For every list of IMP modules, linking their compilations yields the table
obtained by compiling each function of the linked source table.
\end{lemma}

\subsection{Rules for Function Calls}\label{app:sep-call-rules}
A contract $\gamma=\langle\mathsf{pre},\mathsf{post}\rangle$
consists of a precondition $\mathsf{pre}(a,\vec{v})$ over the argument
values and a postcondition $\mathsf{post}(a,r)$ over the result, both
indexed by a witness $a$ of a type that depends on $\gamma$. Contracts that
need no witness use the unit type. Calls are answered by a contract context $\Gamma$, a
partial map from function names to contracts, through the call handler
$\handler{call}(\Gamma)$, which is part of the logical handler of both IMP
and ASM. We write $W_{\vec{a}}(\vec{a},\Phi)$ for the WP of the left-to-right
evaluation of the arguments $\vec{a}$, and reuse the notation
of~\Cref{app:imp-logic} otherwise.
\begin{mathpar}
 \inferrule*[right=Call-Event]
   {\Gamma(f)=\gamma \\ \mathsf{pre}(a,\vec{v}) \\
    \forall r.\;\mathsf{post}(a,r)\wand\Psi(r)}
   {\wpi{H}{\mathsf{FnCall}(f,\vec{v})}{\Psi}}
 \and
 \inferrule*[right=Call]
   {\Gamma(f)=\gamma \\ x\mathrel{\src{\mapsto^{?}}}w \\
    x\mathrel{\src{\mapsto^{?}}}w\wand
      W_{\vec{a}}\bigl(\vec{a},\lambda\vec{v}.\,\mathsf{pre}(a,\vec{v})
        *x\mathrel{\src{\mapsto^{?}}}w\bigr) \\
    \forall r.\;\mathsf{post}(a,r)\wand x\mathrel{\src{\mapsto}}r\wand\Phi(\mathsf{tt})}
   {W(x:=\mathsf{call}\ f(\vec{a}),\Phi)}
\end{mathpar}
\textsc{Call-Event} holds for every handler $H$ that contains
$\handler{call}(\Gamma)$, in particular for the IMP and ASM handlers: a call
event gives away the callee's precondition and receives its postcondition
for the result. \textsc{Call} lifts it to the IMP call command: it evaluates
the arguments, establishes the precondition, and assigns the result to $x$,
which must be owned. In Rocq the postcondition is weakened through an
intermediate predicate; the rule above is the instance where that predicate
is $\mathsf{post}$.

For internal calls between linked modules, the source logic uses
\emph{scoped} contracts (\Cref{sec:secomp-contracts,app:sep-contracts}): the scoped version of
$\gamma$ also lends the local pool at values known to the caller,
$\mathsf{LocalsAt}(W)$ with $\mathrm{dom}\,W=X$, i.e.\ ownership of every
local name in $X$ with its entry in $W$, and returns it unchanged. Calls to
scoped contracts use the following rule, for a local result variable
$x\in X$:
\begin{mathpar}
 \inferrule*[right=Call-Scoped]
   {\Gamma(f)=\mathsf{scoped}_X(\gamma) \\ \mathrm{dom}\,W=X \\ x\in X \\
    \mathsf{LocalsAt}(W) \\
    \mathsf{LocalsAt}(W)\wand
      W_{\vec{a}}\bigl(\vec{a},\lambda\vec{v}.\,\mathsf{pre}(a,\vec{v})
        *\mathsf{LocalsAt}(W)\bigr) \\
    \forall r.\;\mathsf{post}(a,r)\wand
      \mathsf{LocalsAt}(W[x:=\mathsf{Some}(r)])\wand\Phi(\mathsf{tt})}
   {W(x:=\mathsf{call}\ f(\vec{a}),\Phi)}
\end{mathpar}
After the call, the caller gets its locals back with the values they had
before the call, except for $x$, which holds the result. ASM has no rule for the call
instruction beyond \textsc{Call-Event}: target WPs are obtained by
WP-preservation, not proved directly.

\subsection{Contracts and Call Contexts}\label{app:sep-contracts}
A call to $f$ is internal if $f$ belongs to the target table $\trgT$ and
external otherwise. For contexts $\srcGint$ and
$\Gamma_{\mathit{ext}}$, the contexts used by the source and the target
logic are
\[
 \srcGamma(f):=\begin{cases}
   \mathsf{scoped}_X(\srcGint(f)) & f\in \trgT\\
   \Gamma_{\mathit{ext}}(f) & \text{otherwise}
 \end{cases}
 \qquad
 \trgGamma(f):=\begin{cases}
   \compProp{\srcGint(f)} & f\in \trgT\\
   \Gamma_{\mathit{ext}}(f) & \text{otherwise}
 \end{cases}
\]
where, for $\gamma=\langle\mathsf{pre},\mathsf{post}\rangle$,
\[
\begin{aligned}
 \mathsf{LocalsAt}(W)&:=\textstyle\mathop{\scalebox{1.5}{$\ast$}}_{x\mapsto w\in W}
     x\mathrel{\src{\mapsto^{?}}}w,\\
 \mathsf{LocalPool}(X)&:=\exists W.\;\ulcorner\mathrm{dom}\,W=X\urcorner*\mathsf{LocalsAt}(W),\\
 \mathsf{scoped}_X(\gamma)&:=\bigl\langle
   \lambda(W,a)\,\vec{v}.\,\ulcorner\mathrm{dom}\,W=X\urcorner*\mathsf{LocalsAt}(W)*\mathsf{pre}(a,\vec{v}),\\
 &\phantom{{}:=\bigl\langle}
   \lambda(W,a)\,r.\,\mathsf{post}(a,r)*\mathsf{LocalsAt}(W)\bigr\rangle,\\
 \compProp{\gamma}&:=\bigl\langle
   \lambda w\,\vec{v}.\,\Corr*\mathsf{pre}_X(w,\vec{v}),\;
   \lambda w\,r.\,\mathsf{post}_X(w,r)*\Corr\bigr\rangle,
\end{aligned}
\]
with $\langle\mathsf{pre}_X,\mathsf{post}_X\rangle:=\mathsf{scoped}_X(\gamma)$,
and $\Corr=\Corr_{\mathit{Pool},V}$ is defined in~\Cref{app:sep-resources}.
The witness of $\mathsf{scoped}_X(\gamma)$ and of $\compProp{\gamma}$ is the
pair $(W,a)$ of the caller's locals and the witness of $\gamma$.
The pool $X$ must contain the parameters and locals of every linked
function.

Internal calls lend the whole pool because the resolver writes the locals
of the callee when it saves, resets and restores them, and because locals
are shared cells: a callee may overwrite any local of its caller. The pool
comes back with the same values because the resolver restores the
callee's locals to their saved values and the rest of the pool is returned
unchanged by the callee's body, by induction over nested calls. They carry $\Corr$ in the target because the compiled callee needs
the correspondence to replay the simulation for its body. External calls
use $\Gamma_{\mathit{ext}}$ on both sides, and the caller keeps $\Corr$
across them.

\subsection{Resources and State Correspondence}\label{app:sep-resources}
Variables are mapped to addresses by $a_x$. Let $\mathit{Pool}$ be the
union of the registers used by the linked functions, and $V$ a list of
\emph{visible} variables. With $\mathsf{Cell}(x)$ and $\Corr_V$
of~\Cref{fig:imp-asm-simulation},
\[
 \Corr_{\mathit{Pool},V}:=\mathsf{RegOwn}(\mathit{Pool})*\Corr_V.
\]
Compared with $\Corr_c$ of~\Cref{sec:imp-asm-wp-pres},
$\Corr_{\mathit{Pool},V}$ owns the link-wide register pool instead of the
registers of one command, and the cells of the visible variables instead of
those of one command. Like $\Corr_c$, it owns no resource over the whole
store or memory.

\begin{lemma}[Allocation]\label{lem:sep-alloc}
Let $\rho$ be a register file whose domain includes $\mathit{Pool}$,
$\sigma$ a store whose domain includes $V$, and $\mu$ a memory with
$\mu(a_x)=\sigma(x)$ for every $x\in V$. Fresh ghost state yields the
state interpretations of $\rho$ and $\mu$, the correspondence
$\Corr_{\mathit{Pool},V}$, the ownership of the registers of $\rho$
outside $\mathit{Pool}$, a source points-to assertion
$x\mathrel{\src{\mapsto^{?}}}\sigma(x)$ for every variable of $\sigma$, and
target points-to assertions for every cell of $\mu$ that is not the
address of a visible variable.
\end{lemma}
Hence the premises of the theorems below can always be established.

\subsection{The Simulation with Calls}\label{app:sep-simulation}
The event relations of~\Cref{app:imp-asm-simulation} gain a row for calls.
Two call events $\mathsf{FnCall}(f,\vec{v})$ and $\mathsf{FnCall}(f',\vec{v}')$
are related when $f=f'$ and $\vec{v}=\vec{v}'$; at the call, the simulation
hands $\Corr_{\mathit{Pool},V}$ to the callee, and the answers are related
when they are equal and $\Corr_{\mathit{Pool},V}$ is returned. The rows for
store and memory events are those of~\Cref{fig:imp-asm-simulation}, for
the visible variables, with the register pool in place of the registers of
a single command. With these relations, the
three conditions of~\Cref{lem:abstract-wp-transfer} hold as
in~\Cref{sec:wp-pres}, provided that the source and target call handlers
answer related calls compatibly: whenever the target context answers a
call with $\Corr_{\mathit{Pool},V}$ handed over and back, the source
context answers it with the same contract without $\Corr_{\mathit{Pool},V}$.
The contexts $\srcGamma$ and $\trgGamma$ of~\Cref{app:sep-contracts} satisfy
this condition.

\begin{lemma}[Commands with calls]\label{lem:sep-command}
Let $c$ be a command whose registers are in $\mathit{Pool}$ and whose
variables are in $V$, and let the source and target call contexts $\Gamma$
and $\Gamma'$ answer calls compatibly. For every precondition $P$ and
source postcondition ${\src\Phi}$, if
\[
 P\;\vdash\;\wpi{\handler{call}(\Gamma)\oplus\srcH}{\denoteIMP{c}}{{\src\Phi}},
\]
then
\[
 \Corr_{\mathit{Pool},V}*P\;\vdash\;
 \wpi{\trgH_{\Gamma'}}{\denoteASM{\mathsf{compile}(c)}(0)}{\lambda\_.\,\Corr_{\mathit{Pool},V}*{\src\Phi}(\mathsf{tt})},
\]
for every step modality and mask.
\end{lemma}
The proof is that of~\Cref{thm:compiler-wp-pres-full}, by induction on $c$,
with one more case for the call command: its arguments are evaluated as
expressions, and the call events of source and target are related by the
call row.

\subsection{Well-Formedness of Modules}\label{app:sep-wf}
A source function $\mathit{fn}$ with parameters $\vec{p}$, locals
$\vec{l}$, body $c$ and result expression $a$ is \emph{proved against}
$\gamma=\langle\mathsf{pre},\mathsf{post}\rangle$ if, for all witnesses $b$
and argument values $\vec{v}$ of length $|\vec{p}|$, there are $P$ and $\Psi$
with $\mathsf{pre}(b,\vec{v})\vdash P$, $\Psi(r)\vdash\mathsf{post}(b,r)$ for
all $r$, and, for every $W_r$ whose domain is $X$ without $\vec{p}$ and
$\vec{l}$,
\[
\begin{aligned}
 &P*\vec{p}\mathrel{\src{\mapsto}}\vec{v}*\vec{l}\mathrel{\src{\mapsto^{?}}}\mathsf{None}
   *\mathsf{LocalsAt}(W_r)\\
 &\quad\vdash W\bigl(c,\lambda\_.\,
   W_a\bigl(a,\lambda r.\,\Psi(r)*\mathsf{LocalPool}(\vec{p}\vec{l})
     *\mathsf{LocalsAt}(W_r)\bigr)\bigr),
\end{aligned}
\]
where the WPs use the context $\srcGamma$, the step
modality $\mathsf{Later}$ and the empty mask. The body thus gets the rest
$W_r$ of the pool at known values and must return it unchanged, so it can
call functions with other local names. One obligation covers the body and
the result expression, which may read the callee's locals at the values
the body left. A module $M$ is
\emph{well formed} if every function $f$ of $M$ with a contract
$\srcGint(f)=\gamma$ satisfies:
\begin{enumerate}
 \item its parameters and locals are distinct and belong to $X$, and its
       variables belong to $V$;
 \item $\mathsf{pre}(\vec{v})$ implies that $\vec{v}$ has the arity of $f$;
 \item it is proved against $\gamma$.
\end{enumerate}
For linked modules we also require that the source and target tables
contain the functions of $M$ and $\mathsf{compile}(M)$, with distinct names.
The third condition is the only proof obligation; it mentions the other
modules only through their contracts.

\subsection{Proofs of the Main Results}\label{app:sep-proofs}
\begin{lemma}[Per-function open simulation]\label{lem:sep-fun-sim-formal}
Let $f$ have body $\mathit{fn}$ in the source table $\srcT$ and its compiled body
$\mathsf{compile}(\mathit{fn})$ in $\trgT$, with registers in $\mathit{Pool}$
and visible variables. For all argument values $\vec{v}$ of the right
length,
\[
 \Corr\;\vdash\;
 \simIMPASM_R\bigl(\mathsf{run}_{\srcT}(f,\vec{v}),\,\mathsf{run}_{\trgT}(f,\vec{v})\bigr),
 \qquad R(v_1,v_2):=\ulcorner v_1=v_2\urcorner*\Corr.
\]
\end{lemma}
The simulation is open: calls made by $f$ remain events, related when they
have the same callee and arguments.
\Cref{lem:sep-fun-sim-formal} relates $\mathsf{run}_{\srcT}(f,\vec{v})$ and
$\mathsf{run}_{\trgT}(f,\vec{v})$ step by step. Saving, resetting and
restoring the locals are related store and memory events on the cells of
the locals of $f$, which are visible, so $\Corr_V$ owns them. Passing the
arguments is related to the prologue of the compiled function, which
stores the argument registers at the parameter addresses. The body is
related by~\Cref{lem:sep-command}, and the result expression as in the
assignment case of~\Cref{thm:compiler-wp-pres-full}. The pieces are composed
with the bind rule of~\Cref{app:sim-algebra}.

For~\Cref{thm:sep-fun-contract}, the contract proof of $f$ yields a source
WP of $\mathsf{run}_{\srcT}(f,\vec{v})$ under $\srcGamma$, from
$\ulcorner\mathrm{dom}\,W=X\urcorner*\mathsf{LocalsAt}(W)*\mathsf{pre}(b,\vec{v})$
to $\mathsf{post}(b,r)*\mathsf{LocalsAt}(W)$: the pool supplies the cells
that the resolver saves and resets, the rest of the pool is passed to the
body, and the restore returns the callee's locals with their saved
values. \Cref{lem:abstract-wp-transfer}, applied to this WP
and the simulation of~\Cref{lem:sep-fun-sim-formal} with the empty mask on both
sides, gives the target WP under $\trgGamma$. \Cref{cor:sep-module} applies
this to every function of $M$; for the partial resolver, the internal calls
of $\mathsf{run}_{\trgT}(f,\vec{v})$ are re-tagged as internal events, which
does not change their WP.

\subsection{Linked Programs and Adequacy}\label{app:sep-closed}
Item~1 of~\Cref{thm:sep-linking} assumes that $\trgT$ contains exactly the compilation
of each function of $\srcT$, that the variables of every function of $\srcT$
are visible, and that the registers of $c$ are in $\mathit{Pool}$. It holds
for every step modality, with the empty mask. Its proof relies on a linking
law for the simulation.

\begin{lemma}[Linking law]\label{lem:sep-linking-law}
If the source and target \itrees $t_1$ and $t_2$ are related by an open
simulation, in which calls remain events, and the source and target
resolvers are related pointwise, i.e.\ for every call they perform
related computations guarded by a step, then the \itrees obtained by
resolving the calls of $t_1$ and $t_2$ recursively are related by the
closed simulation.
\end{lemma}
\Cref{lem:sep-fun-sim-formal} provides the pointwise relation of the resolvers, the
step performed by each resolver at each call provides the guard, and the
open simulation of the entry command is that of~\Cref{lem:sep-command}. The
linking law relates the two linked programs, and~\Cref{lem:abstract-wp-transfer},
applied once to the closed simulation, yields item~1 of~\Cref{thm:sep-linking}. As
in~\Cref{lem:pointsto-transfer}, points-to assertions of ${\src\Phi}$ on
visible variables can then be extracted to target points-to assertions.

\begin{corollary}[Adequacy]\label{cor:sep-adequacy}
Under the hypotheses of item~1 of~\Cref{thm:sep-linking}, let the linked source
program satisfy ${\src\Phi}$ from $P$, and let $\rho$ and $\mu$ be an
initial register file and memory such that $\Corr_{\mathit{Pool},V}*P$ can
be allocated for them (\Cref{lem:sep-alloc}). Then, for every fuel $n$,
running the compiled linked program for $n$ steps from $\rho$ and $\mu$
under the handler that answers no event satisfies, as an Iris weakest
precondition, the postcondition that holds when the fuel is exhausted,
when the program exits with $\Phi_{\mathit{exit}}$, or when it terminates
in a state in which $\compProp{{\src\Phi}}$ holds. The source and target
logical handlers are thus discharged against the computational
interpretation of their events; a pure statement about final states would
further require the adequacy theorem of Iris, which we have not
mechanized. Since every target exit corresponds to a source undefined
behavior, which the source WP excludes, the corollary also holds with
$\Phi_{\mathit{exit}}:=\mathsf{False}$: the compiled linked program never
exits.
\end{corollary}

\paragraph{Linking with contracts.}
Item~2 of~\Cref{thm:sep-linking} starts from a source proof of the entry
command $c$ in which calls are answered by $\srcGamma$, rather than from a
WP of the linked source program. Its proof applies~\Cref{cor:sep-module} to every module. The resulting
target contracts assume the contracts of internal calls; L\"ob induction
discharges this assumption for all linked functions at once, the step
performed by the resolver at each call supplying the later. This is why
item~2 fixes the step modality to $\mathsf{Later}$. \Cref{lem:sep-command}
gives the WP of the compiled entry command under $\trgGamma$, and resolving
its calls by $\trgT$ preserves it.

\subsection{The Shared Counter in Full}\label{app:sep-counter}
The decrement is
\[
 \mathsf{PRED}(x,k):=k:=0;\ \mathsf{while}\ k+2\le x\ \mathsf{do}\ k:=k+1;\ x:=k,
\]
which satisfies, for $n\neq 0$,
\[
 x\mathrel{\src{\mapsto}}n*k\mathrel{\src{\mapsto^{?}}}w
 \vdash W\bigl(\mathsf{PRED}(x,k),\lambda\_.\,
   x\mathrel{\src{\mapsto}}n{-}1*k\mathrel{\src{\mapsto}}n{-}1\bigr).
\]
With $X:=\{k,r\}$ and $V:=[k;r;G.c]$, the body $c_f$ of $f$ is proved
against its contract as follows: if $G.c$ holds $0$, the postcondition holds
immediately; otherwise $\mathsf{PRED}$ decrements the counter, and
\textsc{Call-Scoped} hands $G.c\mathrel{\src{\mapsto}}n{-}1$ to $g$ and
receives $G.c\mathrel{\src{\mapsto}}0$. The proof for $g$ is symmetric, and
the entry command $\mathit{main}$ is proved in the same way, with a frame
$R$. Linking with contracts then yields, for every $R$,
\[
\begin{aligned}
 &\Corr*\mathsf{LocalPool}(X)*
   (\exists n.\,G.c\mathrel{\src{\mapsto}}n)*R\\
 &\quad\vdash\wpi{\trgH_{\emptyset}}%
   {\mathsf{link}^{\partial}_{\trgT}(\denoteASM{\mathsf{compile}(\mathit{main})}(0))}%
   {{\trg\Psi}},\\
 &{\trg\Psi}:=a_{G.c}\mathrel{\trg{\mapsto}}0
   *\mathsf{LocalPool}(X)*R*\mathsf{RegOwn}(\mathit{Pool})
   *\Corr_{V\setminus\{G.c\}},
\end{aligned}
\]
where $\Corr_{V\setminus\{G.c\}}$ holds the correspondence cells of $k$
and $r$. The precondition owns no target cell other than those of $k$, $r$
and $G.c$.

\end{document}